\documentclass[preprint,11pt]{article}
\usepackage{amsmath}
\usepackage{graphicx}
\usepackage{graphicx, subfigure}
\usepackage{epsfig}
\usepackage[utf8]{inputenc}
\usepackage{floatrow}
\usepackage{float}
\usepackage{amssymb}
\usepackage{amsthm}
\usepackage{epstopdf}
\usepackage{anysize}
\usepackage{authblk}
\usepackage{xpatch}
\usepackage{xcolor}

\usepackage{hyperref}

\marginsize{3cm}{2cm}{2cm}{2cm}

\begin{document}

\title{\bf Immunomodulatory role of black tea in the mitigation of cancer induced by inorganic arsenic}

\author[1]{\small Ravi Kiran\footnote {ravieroy123@gmail.com, ravieroy123@iitkgp.ac.in}}
\author[2]{\small Swati Tyagi}
\author[3]{\small Syed Abbas\footnote {abbas@iitmandi.ac.in}}
\author[4]{\small Madhumita Roy}
\author[1,5]{\small A. Taraphder}

\affil[1]{\footnotesize \textit{Centre for Theoretical Studies,Indian Institute of Technology Kharagpur, Kharagpur 721302, India}}
\affil[2]{\footnotesize \textit{Department of Applied Sciences, Punjab Engineering College, Chandigarh-160012, India}}
\affil[3]{\footnotesize \textit{School of Basic Sciences, IIT Mandi HP 175005} }
\affil[4]{\footnotesize \textit{Department of Environmental Carcinogenesis and Toxicology, Chittaranjan National Cancer Institute, 37 S. P. Mukherjee	Road, Kolkata 700026, India}}
\affil[5]{\footnotesize \textit{Department of Physics, Indian Institute of Technology Kharagpur, Kharagpur 721302, India}}

\vskip .1in
\maketitle
\begin{abstract}
	We present a model analysis of the tumor and normal cell growth under the influence of a carcinogenic agent, an immunomdulator (IM) and variable influx of immune cells including relevant interactions. The tumor growth is facilitated by carcinogens such as inorganic arsenic while the IM considered here is black tea (Camellia sinesnsis). The model with variable influx of immune cells is observed to have considerable advantage over the constant influx model, and while the tumor cell population is greatly mitigated, normal cell population remains above healthy levels. The evolutions of normal and tumor cells are computed from the proposed model and their local stabilities are investigated analytically. Numerical simulations are performed to study the long term dynamics and an estimation of the effects of various factors is made.  This helps in developing a balanced strategy for tumor mitigation without the use of chemotherapeutic drugs that usually have strong side-effects. 
	
\end{abstract}



\newtheorem{theorem}{Theorem}[section]
\newtheorem{lemma}[theorem]{Lemma}
\newtheorem{example}[theorem]{Example}
\newtheorem{proposition}[theorem]{Proposition}
\newtheorem{corollary}[theorem]{Corollary}
\newtheorem{remark}[theorem]{Remark}
\newtheorem{definition}[theorem]{Definition}
\newtheorem{result}[theorem]{Result}


\section{Introduction}

The scourge of cancer is on the rise all around the globe. There have been a number of factors responsible for this increasing trend. Cancer is one of the leading causes of death and according to 2018 GLOBOCAN data, death due to cancer stands at 9.6 million in 2018 \cite{bray2018global}. Cancer is an uncontrolled and abnormal proliferation of cells, leading to formation of tumour which can infiltrate and destroy normal tissues. There have a number of etiologic factors responsible for this disease, including environmental exposures, lifestyle and others. Often there is some occupational exposure to certain carcinogenic risk-factors like chemicals, radioactive materials etc.  

Arsenic, a metalloid, is found in abundance in ground water in many parts of the world. Exposure to arsenic leads to a plethora of health hazards \cite{sinha2010antioxidant}. The United States' Environmental Protection Agency declared that all arsenic is a potential risk to human health \cite{Dibyendu2007arsenic}. In the list of hazardous substances, the United States' Agency for Toxic Substances and Disease Registry ranked arsenic as number 1 \cite{Carelton2007arsenic}. Ingestion of arsenic through drinking water is a global catastrophe and millions of people are affected, particularly in the under-developed world. 

Of the two forms of arsenic, namely organic and inorganic, inorganic Arsenic (iAs) is a potential cause of cancer. There are two inorganic forms of iAs and they are more toxic than the organic forms. These inorganic forms are capable of generation of reactive oxygen species (ROS), leading to DNA, protein and lipid damage. The major cause of arsenic toxicity in humans is through consumption of iAs-contaminated water \cite{mazumder2000diagnosis}. Contamination of ground water with this silent poison is a global problem. 

Administration of cytotoxic drugs, as in chemotherapy, destroys the cancer cells, or slows down their rapid division. Besides killing the cancer cells, these cytotoxic drugs also damage normal cells. Chemotherapeutic agents often lower the immune response, therefore, a combination of chemotherapy with an immune stimulatory agent might be a promising regimen to treat cancer. An immunomodulator may pave the way to a successful treatment strategy. Any agent which boosts immunity with minimal side effects is highly desirable. Natural plant biomolecules are endowed with a number of health-beneficial properties. Many of the natural plant products have enormous anticancer potential. Chemoprevention is therefore a promising strategy that can retard, revert or prevent the progression of cancer and block the development at the initiation stage by use of natural products. 

iAs is a known inducer of ROS, which in turn triggers a number of events promoting carcinogenesis. Therefore, an antioxidant may aid in tackling it. Tea, the most popular beverage around the globe is a good antioxidant. Therefore, studies were undertaken to mitigate carcinogenesis by employing black tea (BT). Among a number of cancers, exposure to iAs causes skin carcinoma. The fact that BT quenches ROS in Swiss albino mice and inhibits iAs-induced skin cancer is already known \cite{sinha2010antioxidant} from experiments carried out on Swiss albino male mice (Mus Musculus) in the laboratory of one of us. A mathematical model was also developed \cite{srivastava2020growth} based on the experimental data.

Besides its antioxidant potential, BT has several other attributes, one such is its immunomodulatory role. Immunotherapy is yet another type of treatment given to patients, which aims to reinforce patients’ own immune response against the growth of cancer cells. The IM regulates the immune function of the normal physiological process and also tries to maintain a healthy immune response in altered physiological conditions. Black tea has been recognized as an immunomodulator in Ayurvedic medicine. The immunomodulatory effects of tea has also been observed in mice \cite{haque2014immunostimulatory, gomes2014black, chattopadhyay2012black}. In this paper, we build upon the previous experimental work and modelling on the effects of BT on tumor growth\cite{srivastava2020growth}. A point of departure from earlier models \cite{de2001mathematical, de2003dynamics} is the inclusion of variable influx of immune cells facilitated primarily by the immunomodulatory effects of BT. Along with the variable influx, the positive interaction between tea and immune cells is also included.


\section{The Model}

In this section we describe the model in detail. As discussed above, we do not use chemotherapy to destroy the cancer cells and instead include the effects of BT in reducing tumor growth (albeit through the reduction of ROS) and the modulation of immunity. 

\subsection{The Model \textemdash Overview}

The model presented has following components:

\begin{enumerate}
	\item \textit{Arsenic as stimulator of tumor:} iAs as an external source to stimulate the growth of tumor. We emphasize that while we consider iAs as driving the tumor growth, based on the experiments discussed above, it can, in general, be viewed as any carcinogen present in the environment to which exposure is common. 
	
	\item \textit{Immune Response:} The model includes immune cells whose growth may be stimulated by the presence of  tumor.  These cells can inhibit the growth of tumor cells. Their effects are taken in the model via the usual kinetic process. The important addition to usual models is a variable influx, s(t) in the evolution of immune cells. This reflects the fact that the immune response is ideally never constant and the body can produce variable amounts of immune cells at different stages of response.

	\item \textit{Competition Terms:} Normal cells and tumor cells compete for available resources, while immune cells and tumor cells compete in predator-prey fashion. 
	
	\item \textit{BT-Response:} BT has mitigating effects on cancer growth (via ROS reduction) and therefore a term representing such effects is included \cite{srivastava2020growth}. The model also includes BT-tumor interaction, as well as BT-immune cell interaction. The beneficial effect of BT is seen through the suppression of  tumor cells and strengthening of the immune cells, and it has no known adverse effects on the normal cells unlike chemotherapy.
	
\end{enumerate}

\subsection{Constructing the Model}

Both the normal and tumor cells independently increase according to the usual logistic growth law. The interaction between normal and tumor cells is of predator-prey type, described by the following system of equations, where normal cells are denoted \footnote{We let $N(t) (\equiv N)$ denote the number of normal cells at time $t$, $T(t) (\equiv T)$ denote the number of tumor cells at time $t$, $I(t) (\equiv I)$ denote the number of immune cells at time $t$, and $s(t)(\equiv s)$ denote the variable influx of immune cells. $A(t) (\equiv A)$ and $D(t) (\equiv D)$ are the Arsenic and BT respectively.} by N and the tumor cells are denoted by T. 

\begin{eqnarray} 
\frac{dN}{dt} &=& r_{2}N\left(1-b_{2}N\right)-c_{1}TN \nonumber\\
\frac{dT}{dt} &=& r_{1}T\left(1-b_{1}T\right)-c_{1}' TN \nonumber
\end{eqnarray}

The interaction terms $c_1$ and $c_1^\prime$ are competition terms and are both assumed to be positive. A negative competition term would imply that instead of the normal cells
destructively competing with the tumor cells for resources and space, the presence of the
normal cells would in fact stimulate further growth of the tumor cell population. While some authors argue that $c_1^\prime$ could be negative \cite{panetta1996mathematical, michelson1996host}, we assume a destructive competition in this study. 

The inclusion of arsenic in the system has deleterious effects on normal cells, and a certain proportion of normal cells become tumorous. It is represented by the following terms

\vspace{0.3cm}

\begin{center}
	$\frac{dN}{dt} = -\alpha' AN $ \hspace{0.5cm} and \hspace{0.5cm} $\frac{dT}{dt} = \alpha AN$
\end{center}

It is possible that not all of the normal cells turn into tumorous, as iAS can also result in the death of normal cells. In earlier studies, a constant influx rate has been assumed for the immune cells. The assumption has been relaxed here to accommodate real immune responses, which is variable, and along with a steady rate, production of immune cells at a rate of $\eta$ is allowed. This is represented by the term

\begin{equation*}
\frac{\eta}{b+s}s
\end{equation*}

To avoid immune cell proliferation and immune-upon-immune crowding, a saturation value is assumed. The presence of tumor cells stimulates the immune response, represented by the positive nonlinear growth term for the immune cells: 

\begin{equation*}
\frac{\rho IT}{\alpha_1 +T}
\end{equation*}

\noindent where $\rho$ and $\alpha_1$ are positive constants. This type of response terms is of the same form as used in the models of Kuznetsov \textit{et.al.} \cite{kuznetsov1994nonlinear}. Moreover, the reaction of immune cells and tumor cells can result in either the death of tumor cells or the deactivation of the immune cells, represented by two competition terms


\begin{center}
	$\frac{dT}{dt} = -\beta IT $ \hspace{0.5cm} and \hspace{0.5cm} $\frac{dI}{dt} = -\beta^\prime IT$
\end{center}

Arsenic and BT considered in the model have a simple form: there is a steady influx of both and a decay with certain rates. The presence of BT is assumed to stimulate the immune response and hence a term similar to Michaelis-Menten equation is assumed. BT also has an adverse effect on the tumor, and both these effects are represented: 

\begin{center}
	$\frac{dT}{dt} = -\gamma DT $ \hspace{0.5cm} and \hspace{0.5cm} $\frac{dI}{dt} = \frac{\delta ID}{\alpha_{2}+D}$
\end{center}

\subsection{The Model \textemdash Equations}

Combining all these terms, we propose and analyze the model described by the following system of equations:
\begin{eqnarray} 
\frac{dN}{dt} &=& r_{2}N\left(1-b_{2}N\right)-c_{1}TN-\alpha' AN\nonumber\\
\frac{dT}{dt} &=& r_{1}T\left(1-b_{1}T\right)-c_{1}' TN+\alpha AN-\beta IT-\gamma DT\nonumber\\
\frac{dI}{dt} &=& s(t)+\frac{\rho IT}{\alpha_1 +T}-d_{I}I-\beta^\prime IT+\frac{\delta ID}{\alpha_{2}+D}\label{eqn1}\\
\frac{dA}{dt} &=& a_0-d_{A}A\nonumber\\
\frac{dD}{dt} &=& b_0-d_{D}D\nonumber\\
\frac{ds}{dt} &=& s_{0}+\frac{\eta}{b+s}s-\mu_{1}s.\nonumber
\end{eqnarray}

\subsubsection{Description of Parameters:}

In this section, we summarize the parameters of the mathematical model presented. This parameter set may vary depending on the case one is analysing. However, the analyses of the model are quite general and hence they still apply.  The idea is to keep the tumor population as low as possible with normal cell count above a healthy threshold. We start with a small amount of tumor cells and immune cells while initial values of arsenic and BT are zero. 

\begin{itemize}
	\item \textit{Per unit growth rates:} $r_1$ and $r_2$ are growth rates for tumor cells and normal cells respectively. Here, we assume the tumor cell population grows more rapidly than the normal cell population, and let $r_1 > r_2.$
	\item \textit{Carrying capacities:} $b_{1}^{-1} \leq b_{2}^{-1} = 1.$
	\item \textit{Competition terms:} $c_1,$ $c_1',$ $ \alpha ,$ $ \alpha' , $ $ \beta $ $ \beta' $ and $ \gamma $ are
	competition terms.
	\item \textit{Death rates:} $d_I,$ $d_A$ and $d_D$ are per capita death rates of immune cells, arsenic, BT respectively; $\mu_1$ is the death rate of the stimulated immune cells. $a_0$ and $b_0$ are constants, assumed 0.4 here. 
	\item \textit{Immune source rate:} The immune source rate is considered to be variable here, denoted by $s(t)$. The
influx of immune cells is to be stimulated from outside; $s_0$ is the constant influx; $ \eta $ is the production rate, $b$ is the saturation term and $\mu_1$ is decay constant.
	\item \textit{Immune response rate:} $ \rho, $ which is assumed to have a baseline value of 1. In the study below, we set $\rho = 0.01$ to simulate a compromised immune system.
	\item \textit{Immune threshold rate:} $\alpha_1,$  is inversely related to the steepness of the immune response curve.
	\item \textit{BT-Immune term:} The BT-immune term is modelled similar to Michaelis-Menten term, with $\alpha_2$ as
	threshold rate and $ \delta $ as response rate.
\end{itemize}

We first analyse the equations analytically for positivity, boundedness, equilibrium points and the stability of the solutions. We also look at possible instabilities. Thereafter, we work out the dynamics and relevant phase diagram numerically.

\section{Positivity and Boundedness of Solutions}\label{sec_pos_bdd}
The system (\ref{eqn1}) has initial conditions given by $ N(0)=N_0\geq 0, $ $ T(0)=T_0\geq 0, $ $ I(0)=I_0\geq 0, $ $ A(0)=A_0\geq 0, $ $ D(0)=D_0\geq 0, $ and $ s(0)=s_0\geq 0. $ We consider all the variables and parameters of the model to be non-negative, since the model investigates cellular populations. Based on the biological findings, we analyze the system (\ref{eqn1}) in the region $ \Omega=\left\{\left(N,T,I,A,D,s\right)\in \mathbb{R}^{6}_{+}\right\}. $ We first assure that the system (\ref{eqn1}) is well posed such that the solutions with non-negative initial conditions remain non-negative for all $ 0<t<\infty $ and thus making the variables biologically meaningful.

\begin{theorem}\label{thm_pos_bdd}
	The region $ \Omega \subset \mathbb{R}^{6}_{+} $ given by $ \displaystyle{\Omega=\left\{\left(N,T,I,A,D,s\right)\in \mathbb{R}^{6}_{+} : N\leq\frac{1}{b_2}\right\}} $ is positively invariant with respect to system of equations (\ref{eqn1}) and non-negative solutions exist for all time $ 0<t<\infty. $
\end{theorem} 
\begin{proof}
	Let $ \Omega \subset \mathbb{R}^{6}_{+} $ given by $ \displaystyle{\Omega=\left\{\left(N,T,I,A,D,s\right)\in \mathbb{R}^{6}_{+} : N\leq\frac{1}{b_2}\right\}} .$ Then the solutions\\ $ \left(N(t),T(t),I(t),A(t),D(t),s(t)\right) $ of (\ref{eqn1}) are positive for all $ t\geq 0.$ It can be observed from the first compartment that 
	\begin{align*}
	\frac{dN}{dt} &= r_{2}N-b_{2}N^{2}-c_{1}TN-\alpha' AN \leq r_{2}N-b_{2}N^{2}.
	\end{align*}
	Using Bernoulli method and considering $ N(0)=N_0, $ we obtain 
	\begin{align}\label{eqn2}
	N(t)&\leq \frac{1}{b_{2}+ke^{-r_{2}t}},
	\end{align}
	where $ \displaystyle{k=\frac{1-N_0 b_2}{N_0}}. $ Thus $ \displaystyle{N_0 = \frac{1}{b_2 +k}.}$
	Substituting the value of $ k $ in (\ref{eqn2}), we obtain
	\begin{align*}
	N(t)&\leq \frac{1}{b_{2}+\frac{1-N_0 b_2}{N_0}e^{-r_{2}t}}\\
	&\leq \frac{1}{b_2}\quad \text{as}\quad  t\rightarrow \infty.
	\end{align*}
	Since $ b_2 >0,$ we have $ N(t)>0 $ for all $ t>0. $
	
	Similarly we can show that $ T(t)>0, $ $ I(t)>0, $ $ A(t)>0, $ $ D(t)>0 $ and $ s(t)>0. $ 
\end{proof}

\section{Equilibrium points}\label{sec_eqbm}
In this section, we discuss the existence of all possible equilibrium points of system (\ref{eqn1}). The model system (\ref{eqn1}) admits three equilibrium points, two dead equilibrium points and one co-existing equilibrium point $ P_*= ( N^*,T^*,I^*,A^*,D^*,s^* )$ respectively. We have $ N^*>0,T^*>0,I^*>0,A^*>0,D^*>0,s^*>0 $ since cell populations are non-negative and real. All parameters are also positive.
\begin{itemize}
%
	\item Dead equilibrium point: An equilibrium point is said to be dead equilibrium, if the normal cell population is zero. For $ \displaystyle{d_I \alpha_2 >(\delta-d_I)\frac{b_0}{d_D}} ,$ we have two dead equilibria in this case. 
	\begin{enumerate}
		\item Type 1 Dead equilibrium point\\
		$ \displaystyle{P_{d1}=(0,0,\frac{s^*(d_D\alpha_2+b_0)}{d_I(d_D\alpha_2+b_0)-b_0\delta },\frac{a_0}{d_A},\frac{b_0}{d_D},s^*)}. $
		\item Type 2 Dead equilibrium point\\
		$ \displaystyle{P_{d2}=(0,m^*,\frac{s^*(d_D\alpha_2+b_0)}{d_I(d_D\alpha_2+b_0)-b_0\delta },\frac{a_0}{d_A},\frac{b_0}{d_D},s^*)}. $
		Here $ m^* $ is the real positive solution of equation $ T^3+AT^2+BT+C=0, $ where
		\begin{align*}
		A &=\frac{1}{\beta' r_1 b_1}\left(-r_1b_1\rho+d_Ir_1b_1+\beta'\alpha_1r_1b_1-\beta'(r_1-\gamma D^*) -\delta D^*r_1b_1\right)\\
		B &= \frac{1}{\beta' r_1 b_1}(s^*\beta(\alpha_2+D^*)+r_1\rho-\rho\gamma D^*+d_I\alpha_1r_1 b_1-d_I(r_1-\gamma D^*)-\beta'\alpha_1r_1+\beta'\alpha_1\gamma D^*\\
		&\quad +\delta D^*(r_1-\gamma D^*)-\delta D^*r_1b_1\alpha_1)\\
		C &= \frac{1}{\beta' r_1 b_1}\left(s^*\beta\alpha_1(\alpha_2+D^*)-d_I\alpha_1 (r_1- \gamma D^*) +\delta D^*\alpha_1(r_1-\gamma D^*)\right)\\
		D^* &= \frac{b_0}{d_D}.
		\end{align*}
	\end{enumerate}
	\item Co-existing equilibrium point $ P_*=\left(N^*,T^*,I^*,A^*,D^*,s^*\right) $ 
\end{itemize}

\section{The Basic Reproduction Number}
In this section, we find the basic reproduction number by following the next generation matrix methods as
described in \cite{castillo2002computation,van2002reproduction}. The matrices $ F $ and $ V $ are given as follows:\\
\begin{center}
	$ F=\begin{pmatrix}
	r_1& 0\\
	0&0
	\end{pmatrix} $
	and $ V=\begin{pmatrix}
	\beta I^*+\gamma D^*& 0\\
	0&d_A
	\end{pmatrix}. $
\end{center}

Then the basic reproduction number $ R_0 $ is given by the largest eigenvalue of $ FV^{-1}. $ Thus, we have 
\begin{equation}\label{eqn3}
R_0 = \frac{r_1d_D}{\beta d_DI^*+\gamma b_0}
\end{equation}
The Basic reproduction number $ R_0 ,$ which measures the rate of spread of tumor; $ R_0 >1 ,$ if each cell produces on average more than one cell and thus the tumor grows over time. If $ R_0<1 ,$ then each cell produces on average less than one new cell and thus the therapy (drug administration) can eradicate the tumor. At each time step, a tumor cell either produces an offspring or dies.


\section{Stability Analysis}

\subsection{Local stability of Dead equilibrium point}
\begin{theorem}
	The Type 1 Dead equilibrium point $ P_{d1}=(0,0,I^*,A^*,D^*,s^*) $ is locally asymptotically stable provided the following holds;
	\begin{equation}\label{dead_las}
	r_2 <\frac{a_0\alpha'}{d_A}, \quad \mu_1 > \frac{b\eta}{(b+s^*)^2}, \quad r_1<\beta I^*+\gamma D^*, \quad d_I > \frac{b_0\delta}{b_0+\alpha_2 d_D}.
	\end{equation}
\end{theorem}
\begin{proof}
	The characteristic equation corresponding to $ P_{d1} $ is given by 
	\begin{equation}\label{eqn5}
	(r_2-\alpha'A^*-\lambda)(r_1-\beta I^*-\gamma D^*-\lambda)(-d_I+\frac{\delta D^*}{\alpha_2+D^*}-\lambda)(-d_A-\lambda)(-d_D-\lambda)(\frac{b\eta}{(b+s^*)^2}-\mu_1-\lambda)=0.
	\end{equation}
	Clearly the roots of characteristic equation (\ref{eqn5}) are $ \displaystyle{r_2-\alpha'A^*<0}, $ $ \displaystyle{r_1-\beta I^*-\gamma D^*<0}, $ $ \displaystyle{-d_I+\frac{\delta D^*}{\alpha_2+D^*}<0}, $ $ \displaystyle{-d_A<0}, $ $ \displaystyle{-d_D<0} $ and $ \displaystyle{\frac{b\eta}{(b+s^*)^2}-\mu_1 <0}. $ Thus the dead equilibrium $ P_{d1} $ is locally asymptotically stable.
\end{proof}

\begin{theorem}
	The Type 2 Dead equilibrium point $ P_{d2}=(0,T^*,I^*,A^*,D^*,s^*) $ is locally asymptotically stable provided (\ref{dead2_las}) holds;
	\begin{equation}\label{dead2_las}
	r_2 <c_1T^*+\alpha' A, \quad \mu_1 > \frac{b\eta}{(b+s^*)^2}.
	\end{equation}
\end{theorem}
\begin{proof}
	The roots of characteristic equation corresponding to equilibrium $ P_{d2} $ are $ \displaystyle{r_2-c_1T^*-\alpha'A^*<0}, $ $ \displaystyle{-d_A<0}, $ $ \displaystyle{-d_D<0} $ and $ \displaystyle{\frac{b\eta}{(b+s^*)^2}-\mu_1 <0}. $ The remaining two roots satisfy the quadratic equation $ \lambda^2-(a+d)\lambda+(ad-bc)=0, $ where $ a=r_1(1-2b_1T^*)-\beta I^*-\gamma D^*, $ $ b=\beta T^*, $ $ \displaystyle{c=\frac{\alpha_1\rho I^*}{(\alpha_1+T^*)^2}-\beta'I^*} $ and $ \displaystyle{d=\frac{\rho T^*}{\alpha_1+T^*}-\beta'T^*-d_I I^*+\frac{\delta D^*}{\alpha_2+D^*}}. $ Clearly the roots are negative if $ a+d<0 $ and $ ad-bc >0 $ along with $ (a-d)^2+4bc > 0. $ Thus the dead equilibrium $ P_{d2} $ is locally asymptotically stable.
\end{proof}

\subsection{Local stability of Endemic equilibrium point}
\begin{theorem}
	The co-existing equilibrium point $ P_*=\left(N^*,T^*,I^*,A^*,D^*,s^*\right) $ is locally asymptotically stable, if the following Routh-Hurwitz criterion is satisfied,
	\begin{align}\label{eqn6_1}
	&A_1+A_4+A_7<0\\
	&(A_1+A_4)(A_1A_4+A_1A_7+A_4A_7+A_{7}^{2}-A_2A_3) +A_5A_6 <0.\nonumber
	\end{align}
	otherwise unstable, where $ A_{i}'s $ are as defined in (\ref{A_i})
\end{theorem}
\begin{proof}
	The characteristic equation of system (\ref{eqn1}) corresponding to $ P_* $ is given by 
	$$ \begin{vmatrix}
	A_1-\lambda& A_3 &0&A_8&0&0\\
	A_2& A_4-\lambda &A_6&A_9&A_{11}&0\\
	0& A_5 &A_7-\lambda&0&A_{12}&1\\
	0& 0 &0&A_{10}-\lambda&0&0\\
	0& 0 &0&0&A_{13}-\lambda&0\\
	0& 0 &0&0&0&A_{14}-\lambda
	\end{vmatrix}=0, $$
	where \begin{align}\label{A_i}
	A_1=r_2(1-2b_2 N^*)-c_1 T^*-\alpha'A^*, &\quad A_2=-c_1' T^*+\alpha A^* \\
	A_3=-c_1 N^*, &\quad A_4= r_1(1-2b_1 T^*)-c_1' N^*-\beta T^* \nonumber\\
	A_5 = \frac{\alpha_1 \rho I^*}{(\alpha_1+T^*)^2}-\beta'I^*, &\quad A_6=-\beta T^*\nonumber\\
	A_7=\frac{\rho T^*}{\alpha_1+T^*}-d_I-\beta'T^*+\frac{\delta D^*}{(\alpha_2+D^*)}, &\quad A_8= -\alpha'N^*, \quad A_9 = \alpha N^*,\nonumber\\
	A_{10}= -d_A,  \quad A_{11}= -\gamma T^*&\quad  A_{12} = \frac{\delta\alpha_2 I^*}{(\alpha_2+D^*)^2} \nonumber\\
	A_{13} = -d_D,& \quad A_{14}=\frac{b\eta}{(b+s^*)^2}-\mu_1.\nonumber
	\end{align}
	Thus we obtain the following equation
	\begin{equation}\label{eqn7}
	(A_{10}-\lambda)(A_{13}-\lambda)(A_{14}-\lambda)\left(\lambda^3+B_1\lambda^2+B_2\lambda+B_3\right)=0,
	\end{equation}
	where $ B_1= -(A_1+A_4+A_7), $ $ B_2 = A_1A_4+A_1A_7+A_4A_7-A_2A_3 $ and $ B_3 = A_5A_6+A_2A_3A_7-A_1A_4A_7. $ Clearly the characteristic equation (\ref{eqn7}) has three roots given by $ \lambda=A_{10}=-d_A <0,$ $\lambda=A_{13}=-d_D <0$ $\displaystyle{\lambda=A_{14}=\frac{b\eta}{(b+s^*)^2}-\mu_1<0}. $ We are left with the cubic equation $ \lambda^3 +B_1\lambda^2+B_2\lambda+B_3 = 0. $ Applying Routh-Hurwitz criterion, the co-existing equilibrium $ P_* $ is locally asymptotically stable, provided $ B_1 >0 $ and $B_1B_2-B_3 >0.$ This implies the co-existing equilibrium $ P_* $ is locally asymptotically stable if $ A_1+A_4+A_7<0 $ and $ (A_1+A_4)(A_1A_4+A_1A_7+A_4A_7+A_{7}^{2}-A_2A_3) +A_5A_6 <0.$
\end{proof}

\subsection{Global Stability of Equilibrium points}

Irrespective of the fact of the stability of the equilibrium points, efforts of doctors have been always oriented to reach to a point where tumor cells are absent. Although in case of Type 1 dead equilibrium, normal cells are also destroyed, but doctors try to protect the immune cells and increase their count. As tumor cells can be completely destroyed at this point, thus resulting in complete therapy of the disease. Moreover, in the mean time, it becomes necessary to find a therapeutic protocol to be able to incline the solution of equations towards this stable point, irrespective of the initial conditions. For this, we need stimulation of immune cells and drug administration, which could guarantee the global stability of this equilibrium point. 

In this section, we employ Lyapunov's direct method \cite{ghaffari} to design the desirable disease eradication protocol. This technique requires selecting a suitable Lyapunov function candidate and then finding a control law to make this candidate a real Lyapunov function. 

\begin{theorem}\label{thm_GAS_TFE}
	The Type 1 steady state, $ \displaystyle{P_{d2} =(0,0,I^*,A^*,D^*,s^*)} $  is globally asymptotically stable if:
\begin{align*}
	(i)& ~ \xi=2\max\{\mu_1,d_A,b_0,\alpha'A^*,\beta I+\gamma D^*,\frac{fs^*}{I^*}\}\\
	(ii)& ~ m< \min\{\frac{c_1'N}{\gamma(D^*-D)}, \frac{c_1gN}{\alpha A}\}
\end{align*}
	for some constants $a,h,c,f, m,g.$
\end{theorem}
\begin{proof}
	We define the following Lyapunov function,
	\begin{align*}
	V(t) &=\frac{a}{2}(s-s^*)^{2}+\frac{h}{2}(A-A^*)^{2}+\frac{c}{2}(D-D^*)^{2}+\frac{g}{2}N^{2}+\frac{f}{2}(I-I^*)^{2}+\frac{m}{2}T^{2},
	\end{align*}
	where $ a, h, c,g,m $ are all positive constants. Computing time derivative of $ V $ along with (\ref{eqn1}), we obtain 
	\begin{align*}
	\frac{dV}{dt} &= a(s-s^*)(s_{0}+\frac{\eta}{b+s}s-\mu_{1}s)+h(A-A^*)(a_0-d_{A}A)+c(D-D^*)(b_0-d_{D}D)\\
	& \quad +gN(r_{2}N\left(1-b_{2}N\right)-c_{1}TN-\alpha' AN)\\& \quad +f(I-I^*)(s(t)+\frac{\rho IT}{\alpha_1 +T}-d_{I}I-\beta IT+\frac{\delta ID}{\alpha_{2}+D})\\ &\quad+mT(r_{1}T\left(1-b_{1}T\right)-c_{1}' TN+\alpha AN-\beta IT-\gamma DT)\\
	&\leq  -\xi V(t)+ITf\beta' I^* - I^2 f\beta' T+mT^2(-c_1'N-\gamma D+\gamma D^*)+I(fs-fs^*)\\&\quad+T(-c_1gN^2+m\alpha AN)-Ag\alpha'N^2-fsI^*+fs^*I^*.
	\end{align*}
	Then using (i), (ii) and (iii), we can obtain $ \displaystyle{\frac{dV}{dt}<0}. $ Thus one can guarantee that solution of equations goes to Type 1 equilibrium point, if the parameters of model system (\ref{eqn1}) satisfy (i)-(iii).
\end{proof}

\begin{corollary}
	The Type 2 steady state, $ \displaystyle{P_* = (0,T^{*},I^{*},A^{*},D^{*},s^{*})}$ is globally asymptotically stable, provided the following conditions hold:
\begin{align*}
	(i)& ~ \xi=2\max\{\mu_1,d_A,b_0,c_1 T^*+\alpha'A^*,\beta I+\gamma D^*+2mr_1b_1T^*,\frac{fs^*}{I^*}\}\\
	(ii) &~\frac{f(s^*-s)}{\beta (T^{*})^2}<m<\min\Big\{\frac{c_1'N}{\gamma(D^*-D)},\frac{c_1gN^2}{(\beta I^*T^*+\alpha AN-c_{1}'NT^*+\gamma DT^*-\gamma D^*T^*)}\Big\}
	\end{align*}
	for some constants $a,h,c,f,m,g.$
\end{corollary}

\begin{corollary}
	The co-existing steady state, $ \displaystyle{P_* = (N^{*},T^{*},I^{*},A^{*},D^{*},s^{*})}$ is globally asymptotically stable, provided the following conditions hold:
	\begin{align*}
	(i)& ~ \xi=2\max\{\mu_1,d_A,b_0,c_1 T^*+\alpha'A^*+2gr_2b_2N^*,c_1 N^*+\beta I+\gamma D^*+2mr_1b_1T^*,\frac{fs^*}{I^*}\}\\
	(ii)& ~\frac{m\beta (T^{*})^2}{s^*-s+\beta'I^*T^*}<f < \min\Big\{\frac{m\beta T^*}{\beta'I^*},\frac{mT^*(\beta I^*T^*-\alpha N^*A^*)+c_1 gNN^*T^*}{I^*(s^*-s)}\Big\},\\
	(iii)& ~ \frac{g\alpha'(N^*-N)}{T^*} < m <\frac{c_1 gN(N-N^*)}{(\beta I^*T^*+\alpha AN-\alpha A^*N^*-c_{1}'NT^*-c_{1}'N^*T^*+\gamma DT^*-\gamma D^*T^*)} \\
	\end{align*}
	for some constants $a,h,c,f,m,g.$
\end{corollary}

\section{Delayed Model}\label{sec_Delayed}

In comparison to non-delayed models, delay differential equation (DDEs) systems can exhibit much richer dynamics since
a time delay could cause the loss of stability of equilibrium and give rise to periodic solutions through the Hopf bifurcation. The instabilities and oscillatory behavior caused by delays are very common, however the delays may also have the opposite effect, namely that they can suppress oscillations and stabilize equilibria which would be unstable in the absence of delays. 

Due to chemical transportation of signals and the time needed for differentiation/division of cells, the production of tumor cells may not be instantaneous but, instead, it exhibits some time lag. Moreover, due to the immunity, there may be time delay in competition of normal cells and tumor cells as well. To capture such phenomenon, we introduce two delays into the non-delayed model (\ref{eqn1}). Thus we obtain the following system of DDEs corresponding to system (\ref{eqn1}) described by the following equations;
\begin{eqnarray} 
\frac{dN}{dt} &=& r_{2}N\left(1-b_{2}N\right)-c_{1}TN-\alpha' A(t-\tau_2)N(t-\tau_2)\nonumber\\
\frac{dT}{dt} &=& r_{1}T\left(1-b_{1}T\right)-c_{1}' TN-\beta IT-\gamma DT+\alpha A(t-\tau_1)N(t-\tau_1)\nonumber\\
\frac{dI}{dt} &=& s(t)+\frac{\rho IT}{\alpha_1 +T}-d_{I}I-\beta IT+\frac{\delta ID}{\alpha_{2}+D}\label{eqn_delay}\\
\frac{dA}{dt} &=& 0.4-d_{A}A\nonumber\\
\frac{dD}{dt} &=& 0.4-d_{D}D\nonumber\\
\frac{ds}{dt} &=& s_{0}+\frac{\eta}{b+s}s-\mu_{1}s.\nonumber
\end{eqnarray}
We denote by $ C $ the Banach space of continuous functions $ \phi:[-\tau,0]\rightarrow R^6 $ equipped with suitable norm, where $ \tau =\max\{\tau_1, \tau_2\}. $ Further let $ C_+=\{\phi=(\phi_1,\phi_2,\phi_3,\phi_4,\phi_5,\phi_6)\in C:\phi_{i}(\theta)\geq 0 ~\forall~ \theta\in [-\tau,0], $ $ i=1,2,\cdots,6. $ The initial conditions corresponding to delayed system (\ref{eqn_delay}) are 
\begin{equation}\label{eqn8}
N(\theta)=\phi_{1}(\theta), ~~T(\theta)=\phi_{2}(\theta), ~I(\theta)=\phi_{3}(\theta), ~A(\theta)=\phi_{4}(\theta), ~D(\theta)=\phi_{5}(\theta), ~s(\theta)=\phi_{6}(\theta),
\end{equation}
where $ \phi=(\phi_1,\phi_2,\phi_3,\phi_4,\phi_5,\phi_6)\in C_{+}. $

\subsection*{Qualitative Analysis: Preliminaries}
In this subsection, we establish the non-negativity of the solutions of system (\ref{eqn_delay}) with initial conditions (\ref{eqn8}).

\begin{proposition}
	The solutions $ \left(N(t),T(t),I(t),A(t),D(t),s(t)\right) $ with initial conditions (\ref{eqn8}) of the system (\ref{eqn_delay}) are non-negative.
\end{proposition}
\begin{proof}
	We can rewrite the system (\ref{eqn_delay}) in vector form by setting $ Z=\left(N,T,I,A,D,s\right)^T \in R^6 $ and 
	\begin{equation}\label{eqn9}
	F(Z) = \begin{pmatrix}
	F_{1}(Z)\\F_{2}(Z)\\F_{3}(Z)\\F_{4}(Z)\\F_{5}(Z)\\F_{6}(Z)
	\end{pmatrix} = \begin{pmatrix}
	r_{2}N\left(1-b_{2}N\right)-c_{1}TN-\alpha' A(t-\tau_2)N(t-\tau_2)\\
	r_{1}T\left(1-b_{1}T\right)-c_{1}' TN-\beta IT-\gamma DT+\alpha A(t-\tau_1)N(t-\tau_1)\\
	s(t)+\frac{\rho IT}{\alpha_1 +T}-d_{I}I-\beta IT+\frac{\delta ID}{\alpha_{2}+D}\\
	0.4-d_{A}A\\
	0.4-d_{D}D\\
	s_{0}+\frac{\eta}{b+s}s-\mu_{1}s.
	\end{pmatrix}
	\end{equation}
	where $ F:C_{+}\rightarrow \mathbb{R}^{6} $ and $ F\in C^{\infty}(\mathbb{R}^{6}). $ Then delayed model (\ref{eqn_delay}) becomes
	\begin{equation}\label{eqn10}
	\dot{Z}(t)=F(Z_t),
	\end{equation}
	where $ \dot \equiv \frac{d}{dt}, $ $ Z_t(\theta)=Z(t+\theta), $ $ \theta \in [-\tau,0]. $ It can be observed from (\ref{eqn10}) that whenever we choose $ Z(\theta) \in C_+ $ such that $ Z_i = 0, $ we obtain $ F_i(z)|_{Z_{i}(t)=0}, $ $ Z_{+}\in C_{+} \geq 0, $ $ i=1,2,\cdots,6. $ Then from \cite{yang1996permanence}, any solution (\ref{eqn10}) with $ Z_{t}(\theta)\in C_+ $ say $ Z(t)=Z(t, Z(0)) $ is such that $ Z(t)\in R^{6}_{0^{+}} ~\forall ~t>0. $ Further we define 
	\begin{equation*}
	P(t)=N(t)+T(t)+I(t)+A(t)+D(t)+s(t).
	\end{equation*}
	Then using the similar analysis as done in Theorem \ref{thm_pos_bdd}, we have $ P(t) $ is bounded and hence are $N(t),T(t),I(t),A(t),D(t),s(t).  $ This completes the proof.
\end{proof}

\subsection*{Stability and Bifurcation Analysis}\label{subsec_stab_bif}
To investigate the local stability of equilibria of system (\ref{eqn_delay}), we linearize the system and evaluate the characteristic equation first at equilibrium $ P_* .$ The characteristic equation is 
$$ \begin{vmatrix}
A_1-\lambda& A_3 &0&A_8&0&0\\
A_2& A_4-\lambda &A_6&A_9&A_{11}&0\\
0& A_5 &A_7-\lambda&0&A_{12}&1\\
0& 0 &0&A_{10}-\lambda&0&0\\
0& 0 &0&0&A_{13}-\lambda&0\\
0& 0 &0&0&0&A_{14}-\lambda
\end{vmatrix}=0, $$
where \begin{align}\label{delay_A_i}
&A_1=r_2(1-2b_2 N^*)-c_1 T^*-\alpha'Ae^{-\lambda\tau_2}=P-\alpha'Ae^{-\lambda\tau_2},\nonumber\\ 
& A_2=-c_1' T^*+\alpha Ae^{-\lambda\tau_1}=Q+\alpha Ae^{-\lambda\tau_1} \\
&A_3=-c_1 N^*, \quad A_4= r_1(1-2b_1 T^*)-c_1' N^*-\beta T^* \nonumber\\
&A_5 = \frac{\alpha_1 \rho I^*}{(\alpha_1+T^*)^2}-\beta'I^*, \quad A_6=-\beta T^*\nonumber\\
&A_7=\frac{\rho T^*}{\alpha_1+T^*}-d_I-\beta'T^*+\frac{\delta D^*}{(\alpha_2+D^*)}, \nonumber\\
& A_8= -\alpha'Ne^{-\lambda\tau_2}, \quad A_9 = \alpha Ne^{-\lambda\tau_1},\nonumber\\
&A_{10}= -d_A,  \quad A_{11}= -\gamma T^* \quad  A_{12} = \frac{\delta\alpha_2 I^*}{(\alpha_2+D^*)^2} \nonumber\\
&A_{13} = -d_D, \quad A_{14}=\frac{b\eta}{(b+s^*)^2}-\mu_1.&
\end{align}
The characteristic equation is
\begin{equation}
P_{0}(\lambda)+P_{1}(\lambda)e^{-\lambda\tau_1}+P_{2}(\lambda)e^{-\lambda\tau_2}=0,
\end{equation}
where $ \lambda $ is an eigenvalue and 
\begin{align*}
P_{0}(\lambda)&=\lambda^3+\lambda^2(-P-A_4 -A_7)+\lambda(A_4 P+A_7 P+A_4A_7-QA_3)+(A_5A_6+QA_3A_7-PA_4A_7)\\&\quad =\lambda^3+H_0\lambda^2+H_1\lambda+H_2\\
P_{1}(\lambda)&=-\alpha A A_3\lambda+\alpha A_3A_7=H_3\lambda+H_4\\
P_{2}(\lambda)&=\alpha' A\lambda^2+\alpha' AA_4A_7=H_5\lambda^2+H_6.
\end{align*}

\textbf{Case 1: $\tau_1=\tau_2=0$}\\
This case is equivalent to (\ref{eqn7}) of non-delayed model. Thus $ P_* $ is locally asymptotically stable provided (\ref{eqn6_1}) holds. Similarly, Type 1 and Type 2 dead equilibrium are locally asymptotically stable, if (\ref{dead_las}) and (\ref{dead2_las}) hold.

\textbf{Case 2: $\tau_1 >0, \tau_2 =0$}\\
In this case, the characteristic equation becomes $ P_{0}(\lambda)+P_{1}(\lambda)e^{-\lambda\tau_1} =0. $ We can rewrite this equation as 
\begin{equation}\label{eqn11}
\left(\lambda^3+H_0\lambda^2+H_1\lambda+H_2\right)+\left(H_3\lambda+H_4\right)e^{-\lambda\tau_1}=0.
\end{equation}
If the time delay $ \tau_1 $ is able to destablize $ P_* $ and produces oscillations, this can occur only when characteristic roots cross the imaginary axis to the right. Let $ \omega >0 $ and let $ \lambda=i\omega $ is a purely imaginary root of (\ref{eqn11}). Separating real and imaginary parts, we have
\begin{align}\label{eqn12}
H_4\cos\omega\tau_1  +H_{3}\omega\sin\omega\tau_1 &=H_{0}\omega^2-H_2\\
H_{3}\omega\cos\omega\tau_1-H_4\sin\omega\tau_1 &=\omega^3-H_1\omega.\label{eqn12_2}
\end{align}
Eliminating $ \tau_1 $ in (\ref{eqn12}) and (\ref{eqn12_2}), we obtain the following sixth-degree polynomial equation
\begin{equation}\label{eqn13}
F(\omega)=\omega^6 +h_2\omega^4+h_1\omega^2+h_0 =0,
\end{equation}
where $ h_2=H_{0}^{2}-2H_1, $ $ h_{1}=H_{1}^{2}-2H_0H_2-H_{3}^{2} $ and $ h_{0}=H_{2}^{2}-H_{4}^{2}. $ Letting $ v=\omega^2 $ gives the following simplified system
\begin{equation}\label{eqn14}
G(v)=v^3 +h_2 v^2+h_1 v+h_0 =0.
\end{equation}
Define $ \Delta = h_{2}^{2}-3h_1. $ Next we follow the method described in \cite{song2006bifurcation} to investigate for the existence of positive roots of equation (\ref{eqn13}). From (\ref{eqn14}), we have
\begin{equation}\label{eqn15}
\frac{dG(v)}{dv} = 3v^2 +2h_2 v+h_1
\end{equation}
If $ h_{0}<0, $ since $ \lim_{v\rightarrow\infty}G(v)=\infty, $ then we obtain that (\ref{eqn13}) has atleast one positive root.

When $ \Delta \leq 0, $ $ \displaystyle{\frac{dG(v)}{dv}\geq 0}, $ that is $ G(v) $ is monotonically increasing. Thus, if $ h_{0}\geq 0 $ and $ \Delta \leq 0, $ then (\ref{eqn13}) has no positive root.

When $ \Delta >0, $ then graph of $ G(v) $ has two critical points $ \displaystyle{v_{1}^{*}=\frac{-h_{2}+\sqrt{\Delta}}{3}}, $ $ \displaystyle{v_{2}^{*}=\frac{-h_{2}-\sqrt{\Delta}}{3}}. $  Therefore, if $ h_{0}\geq 0, $ then from Lemma 2.2 of \cite{song2006bifurcation}, we can say that (\ref{eqn15}) has positive roots if and only if $ \Delta >0, $ $ v_{1}^{*}>0 $ and $ G(v_{1}^{*}) \leq 0. $ Assume that equation (\ref{eqn15}) has three positive roots, given by $ v_1, $ $ v_2, $ $ v_3 $ respectively. Then (\ref{eqn15}) has three positive roots $ \omega_k=\sqrt{v_{k}}, $ $ k=1,2,3. $ Solving (\ref{eqn12}) and (\ref{eqn12_2}) for $ \tau_1, $ we obtain
\begin{equation}\label{eqn16}
\tau_{1,k}^{(n)} =\frac{1}{\omega_k}\arccos \frac{H_4(H_0 \omega^2)-H_2)+H_3\omega^2(\omega^2 -H_1)}{2(H_{4}^{2}+\omega^2 H_{3}^{2})}+\frac{2n\pi}{\omega_k}, \quad k=1,2,3,\cdots, \quad n=0,1,2,\cdots
\end{equation}
and $ \pm i\omega_k $ is pair of purely imaginary roots of (\ref{eqn13}) with $ \tau_{1,k}^{(n)}. $

We further define 
\begin{equation}\label{eqn17}
\tau_{1}^{*}=\tau_{1,k_0}^{(0)}=\min_{k=\{1,2,3\}}\left\{\tau_{1,k}^{(0)}\right\}, \quad \omega^*=\omega_{k_0}.
\end{equation}
We now obtain the transversality condition for Hopf Bifurcation at $ \tau_1 =\tau_{1}^{*}. $ Differentiating (\ref{eqn11}) with respect to $ \tau_1 $ and substituting expression for $ e^{-\lambda\tau_1} $ from (\ref{eqn11}), we obtain
\begin{align*}
\frac{d\lambda}{d\tau}\Big[3\lambda^2 +2\lambda H_0 +H_1-\frac{(H_3-H_3\lambda\tau_1 -H_4\tau_1)(\lambda^3+H_0\lambda^2+H_1 \lambda+H_2)}{H_3\lambda+H4}\Big]= -\lambda(\lambda^3+H_0\lambda^2+H_1 \lambda+H_0)
\end{align*}
Thus we have
\begin{align*}
\left(\frac{d\lambda}{d\tau}\right)^{-1}&=\frac{H_3}{\lambda(H_3\lambda+H4)}-\frac{3\lambda^2 +2\lambda H_0 +H_1}{\lambda(\lambda^3+H_0\lambda^2+H_1 \lambda+H_2)}-\frac{\tau_1}{\lambda}.
\end{align*}
Evaluating $ \left(\frac{d\lambda}{d\tau}\right)^{-1} $ at $ \tau_1=\tau_{1}^{*}(i.e. \lambda =i\omega^{*}) $ and taking real part, we obtain
\begin{align*}
Re\left[\left(\frac{d\lambda}{d\tau}\right)^{-1}\Big|_{\tau_1 = \tau_{1}^{*}}\right]&=\frac{3\omega^4+2\omega^2(H_{0}^{2}-2H_1)+(H_{1}^{2}-2H_{0}H_{2}-H_{3}^{2})}{H_{3}^{2}\omega^2 +H_{4}^{2}}=\frac{G'(\omega^{*2})}{H_{3}^{2}\omega^{*2} +H_{4}^{2}}
\end{align*}We have $ G(v) $ is non-increasing and positive. Thus
\begin{equation*}
sign\left\{\left(\frac{dRe(\lambda)}{d\tau}\right)\Big|_{\tau_1 = \tau_{1}^{*}}\right\}=sign\left\{\left(\frac{d\lambda}{d\tau}\right)^{-1}\Big|_{\tau_1 = \tau_{1}^{*}}\right\}=sign\{G'(\omega^{*2})\}.
\end{equation*}
Hence if $ G'(\omega^{*2}) \neq 0, $ then transversality condition holds. Summarizing the above, we have the following theorem

\begin{theorem}
	For $ \tau_1>0, $ $ \tau_2 =0, $ assume that condition (\ref{eqn6_1}) holds. If either $ h_{0}<0 $ or $ h_0 \geq 0, $ $ \Delta >0, $ $ v_{1}^{*}>0 $ and $ G(v_{1}^{*}) <0, $ then the Endemic equilibrium $ P_* $ is locally asymptotically stable for $ 0<\tau_1<\tau_{1}^{*}, $ where
	\begin{equation}\label{eqn18}
	\tau_{1}^{*} =\frac{1}{\omega^*}\arccos \frac{H_4(H_0 \omega^2)-H_2)+H_3\omega^2(\omega^2 -H_1)}{2(H_{4}^{2}+\omega^2 H_{3}^{2})}+\frac{2n\pi}{\omega^*}.
	\end{equation}
	Furthermore, if $ G'(\omega^{*2}) \neq 0, $ then system (\ref{eqn_delay}) undergoes Hopf bifurcation to periodic solutions at $ P_* $ at $ \tau_1 = \tau_{1}^{*}. $
\end{theorem}


\begin{remark}
	If $ h_0 \geq 0 $ and $ \Delta \leq 0, $ then equation (\ref{eqn15}) has no positive real root, thus the equilibrium $ P_* $ is locally asymptotically stable for all $ \tau_1 >0. $
\end{remark}

\textbf{Case 3: $ \tau_2 >0, \tau_1 =0 $}\\ 
In this case, the characteristic equation becomes $ P_{0}(\lambda)+P_{1}(\lambda)e^{-\lambda\tau_2} =0. $ We can rewrite this equation as 
\begin{equation}\label{eqn19}
\left(\lambda^3+H_0\lambda^2+H_1\lambda+H_2\right)+\left(H_5\lambda^{2}+H_6\right)e^{-\lambda\tau_2}=0.
\end{equation}
Using the similar analysis as done in case 2, we have the following theorem.
\begin{theorem}
	For $ \tau_2>0, $ $ \tau_1 =0, $ assume that condition (\ref{eqn6_1}) holds. If either $ h_{0}<0 $ or $ h_0 \geq 0, $ $ \Delta >0, $ $ v_{1}^{*}>0 $ and $ G(v_{1}^{*}) <0, $ then the Endemic equilibrium $ P_* $ is locally asymptotically stable for $ 0<\tau_2<\tau_{2}^{*}, $ where
	\begin{equation}\label{eqn20}
	\tau_{2}^{*} =\frac{1}{\omega^{*}}\arccos \frac{(H_0 \omega^2-H_2)}{2(H_{6}-\omega^2 H_{5})}+\frac{2n\pi}{\omega^{*}}
	\end{equation}
	Furthermore, if $ G'(\omega^{*2}) \neq 0, $ then system (\ref{eqn_delay}) undergoes Hopf bifurcation to periodic solutions at $ P_* $ at $ \tau_2 = \tau_{2}^{*}. $
\end{theorem}

\textbf{Case 4: $ \tau_2 >0, \tau_1 \in (0,\tau_{1}^{*}) $}\\
In this case, we consider $ \tau_2 $ as a parameter and fix $ \tau_1 $ at a point in its stable interval. At $ P_* ,$ the characteristic equation takes the following form;
\begin{equation}\label{eqn21}
P(\lambda,\tau_1,\tau_2)=\left(\lambda^3+H_0\lambda^2+H_1\lambda+H_2\right)+\left(H_3\lambda+H_4\right)e^{-\lambda\tau_1}+\left(H_5\lambda^2+H_6\right)e^{-\lambda\tau_2}=0.
\end{equation}
Assume that (\ref{eqn21}) has purely imaginary root given by $ \lambda=i\omega. $ Substituting in (\ref{eqn21}) and separating real and imaginary parts, we obtain
\begin{align}\label{eqn22}
H_{0}\omega^2 -H_2 -H_4\cos\omega\tau_1 -H_3\omega\sin\omega\tau_1 &=(H_6 - H_5\omega^2)\cos\omega\tau_2\\
-\omega^3 +H_1\omega +H_3\omega\cos\omega\tau_1 -H_4\sin\omega\tau_1 &=(H_6 - H_5\omega^2)\sin\omega\tau_2.\label{eqn23}
\end{align}
Eliminating $ \tau_2 $ in (\ref{eqn22}) and (\ref{eqn23}), we obtain the following polynomial equation;
\begin{align}\label{eqn24}
&\omega^6 +\omega^4(H_{0}^{2}-2H_1 - H_{5}^{2})+\omega^3(2H_4\sin\omega\tau_1-2H_0H_3\sin\omega\tau_1)\nonumber\\
&\quad+\omega^2(H_{3}^{2} -2H_0H_2+H_{1}^{2}+2H_5H_6-2H_1H_3\cos\omega\tau_1)=0,
\end{align}
which is a sixth-degree polynomial in $ \omega. $ If (C1) holds, then let we denote the six positive roots of (\ref{eqn24}) as $ \omega_k, $ $ k=1,2,\cdots,6. $ Solving (\ref{eqn22})-(\ref{eqn23}) for $ \tau_2, $ we obtain
\begin{equation}\label{eqn25}
\tau_{2,k}^{n}=\frac{1}{\omega_k}\arccos\frac{H_{0}\omega^2 -H_2 -H_4\cos\omega\tau_1 -H_3\omega\sin\omega\tau_1}{H_6 - H_5\omega^2}+\frac{2n\pi}{\omega_k},\quad k=1,2,\cdots,6\quad n=0,1,2\cdots
\end{equation} 
and $ \pm\omega_k $ is pair of purely imaginary root of (\ref{eqn24}) with $ \tau_{2,k}^{n}. $ For simplicity we define
\begin{equation}\label{eqn26}
\tau_{20}^{*}=\tau_{2,k}^{0}=\min_{k=1,2,\cdots,6}\left\{\tau_{2,k}^{0}\right\},\quad \omega^*=\omega_{k_0}.
\end{equation}
Now, differentiating (\ref{eqn21}) with respect to $ \tau_2 $ and further simplifying for transversality condition, we obtain
\begin{align*}
Re\left[\left(\frac{d\lambda}{d\tau}\right)^{-1}\Big|_{\tau_1 = \tau_{1}^{*}}\right]&= \frac{2\omega H_0\cos\omega\tau_2 +(H_1 -3\omega^2)\sin\omega\tau_2}{\omega(H_6 -H_5\omega^2)}+\frac{H_3\sin\omega(\tau_2 -\tau_1)}{\omega(H_6 -H_5\omega^2)}+\frac{2H_5}{(H_6 -H_5\omega^2)}\Big|_{\tau_{20}^{*}}.
\end{align*}
Since (C2) holds, the transversality condition holds. Thus we have the following result.
\begin{theorem}
	Assume condition (C1)-(C2) hold and $ \tau_1 \in (0,\tau_{1}^{*}), $ then $ P_* $ is locally asymptotically stable for $ \tau_2 \in (0,\tau_{20}^{*}) $ and system (\ref{eqn_delay}) undergoes Hopf bifurcation to periodic solutions at $ P_* $ for $ \tau_2 =\tau_{20}^{*} ,$ where\\
	(C1): Equation (\ref{eqn24}) has six positive roots,\\
	(C2): $ H_6 > H_{5}\omega^2\Big|_{\tau_{20}^{*}}, $ $ H_1 > 3\omega^2\Big|_{\tau_{20}^{*}}. $
\end{theorem}

\textbf{Case 5: $ \tau_1>0, \tau_2 \in (0,\tau_{2}^{*}) $}\\
In this case, $ \tau_2 $ is fixed at a point in its stable interval and $ \tau_1 $ is considered as a parameter. At $ P_* $ we have the same characteristic equation as (\ref{eqn21}). Following the similar procedure as in case 4, we substitute $ \lambda=i\omega $ and separate the real and imaginary parts. Thus we obtain
\begin{align}
H_4 \cos\omega\tau_1 +H_3\omega\sin\omega\tau_1 & = H_{0}\omega^2 -H_2-(H_6 -H_5\omega^2)\cos\omega\tau_2\\
H_3\omega\cos\omega\tau_1-H_4 \sin\omega\tau_1 & = \omega^3 -H_1\omega+(H_6 -H_5\omega^2)\sin\omega\tau_2.\nonumber
\end{align}
Thus we obtain
\begin{align}\label{eqn27}
&\omega^6-(2H_5\sin\omega\tau_2)\omega^5+(H_{0}^{2}+H_{5}^{2}+2H_0H_5\cos\omega\tau_2)\omega^4+(2H_1H_5\sin\omega\tau_2+2H_6\sin\omega\tau_2)\omega^3\\
&\quad+(-2H_5H_6-2H_0H_2-2H_2H_5\cos\omega\tau_2-2H_0H_6\cos\omega\tau_2+H_{1}^{2}-H_{3}^{2})\omega^2-(2H_1H_6\sin\omega\tau_2)\omega\nonumber\\
&\quad+(H_{2}^{2}+H_{6}^{2}-2H_2H_6\cos\omega\tau_2 -H_{4}^{2}) =0
\end{align}
and 
\small{\begin{align}\label{eqn28}
	&\tau_{1,k}^{n}=\frac{1}{\omega_k}\arccos\frac{H_4(H_{0}\omega^2 -H_2-(H_6 -H_5\omega^2)\cos\omega\tau_2)+H_3\omega(\omega^3 -H_1\omega+(H_6 -H_5\omega^2)\sin\omega\tau_2)}{2(H_{4}^{2}+H_{3}^{2}\omega^2)}+\frac{2n\pi}{\omega_k},\\&\quad k=1,2,\cdots,6\quad n=0,1,2\cdots.\nonumber
	\end{align} }
Furthermore, we have
\small{\begin{align*}
	&sign\Big[Re\left[\left(\frac{d\lambda}{d\tau}\right)^{-1}\Big|_{\tau_1 = \tau_{1}^{*}}\right]\Big]\\&= sign\Big[\frac{\cos\omega\tau_1(3H_3\omega^3-H_1H_3\omega +2H_0H_4\omega)-\sin\omega\tau_1(3\omega^3 H_4 -H_1H_4-2H_0H_3\omega^2)+H_3 -2H_5\omega\sin\omega(\tau_1 -\tau_2)}{\omega(H_{4}^{2}+H_{3}^{2}\omega^2)}\Big].
	\end{align*}}

\begin{theorem}
	Assume conditions (C1')-(C2') hold and $ \tau_2 \in (0,\tau_{2}^{*}), $ then $ P_* $ is locally asymptotically stable for $ \tau_1 \in (0,\tau_{10}^{*}) $ and system (\ref{eqn_delay}) undergoes Hopf bifurcation to periodic solutions at $ P_* $ for $ \tau_1 =\tau_{10}^{*} ,$ where\\
	(C1'): Equation (\ref{eqn27}) has six positive roots,\\
	(C2'): $ 3H_3\omega^3 +2\omega H_0H_4>H_1H_3\omega, $ $ 3\omega^2H_4<H_1H_4+2H_0H_3\omega^2 $ and $ 2H_5\omega\sin\omega(\tau_1 -\tau_2)<0. $
\end{theorem}

\textbf{Case 6: $ \tau_1=\tau_2=\tau $}\\
In this case, characteristic equation is 
\begin{equation}\label{eqn29}
P(\lambda,\tau)=\left(\lambda^3+H_0\lambda^2+H_1\lambda+H_2\right)+\left(H_5\lambda^2+H_3\lambda+H_7\right)e^{-\lambda\tau}=0,
\end{equation}
where $ H_7 = H_4+H_6. $ Then substituting $ \lambda=\i\omega $ in (\ref{eqn29}) and following the similar procedure, we obtain
\begin{equation}\label{eqn30}
\tau_{k}^{i}=\frac{1}{\omega_k}\arccos\frac{H_3(H_0\omega^2 -H_2)+(H_7-H_5\omega^2)(\omega^2 -H_1)}{2H_3(H_7-H_5\omega^2)}+\frac{2n\pi}{\omega_k},\quad k=1,2,\cdots,6\quad n=0,1,2\cdots.
\end{equation}
and 
\begin{align}\label{eqn31}
\omega^6+(H_{0}^{2}-H_{5}^{2}-2H_1)\omega^4+(H_{1}^{2}-2H_0H_2+2H_5H_7)+(H_{2}^{2}-H_{7}^{2}) =0.
\end{align}
Furthermore, we have
{\scriptsize{\begin{align*}
		&sign\left[Re\left[\left(\frac{d\lambda}{d\tau}\right)^{-1}\Big|_{\tau = \tau^{*}}\right]\right]\\&= sign\left[\frac{(H_7 -H_5\omega^2)\left(2H_0\omega\cos\omega\tau +(H_1 -3\omega^2)\sin\omega\tau\right)-H_3\omega\left((H_1 -3\omega^2)\cos\omega\tau-2H_0\omega\sin\omega\tau\right)+2\omega H_5(H_7-H_5\omega^2)-H_{3}^{2}\omega}{\omega\left((H_5\omega^2 -H_7)^2+H_{3}^{2}\omega^2\right)}\right].
		\end{align*}}}

\begin{theorem}
	Assume conditions (C1'')-(C2'') hold and $ \tau \in (0,\tau^{*}), $ then $ P_* $ is locally asymptotically stable for $ \tau \in (0,\tau^{*}) $ and system (\ref{eqn_delay}) undergoes Hopf bifurcation to periodic solutions at $ P_* $ for $ \tau =\tau^{*} ,$ where\\
	(C1''): Equation (\ref{eqn29}) has six positive roots,\\
	(C2''): \footnotesize{$ (H_7 -H_5\omega^2)\left(2H_0\omega\cos\omega\tau +(H_1 -3\omega^2)\sin\omega\tau\right)+2\omega H_5(H_7-H_5\omega^2)>H_3\omega\left((H_1 -3\omega^2)\cos\omega\tau-2H_0\omega\sin\omega\tau\right)+H_{3}^{2}\omega. $}
\end{theorem}
\newpage
\section{Numerical Simulations}\label{sec_numerical}

\subsection{Original Model}\label{sec: non-delay}

\noindent We begin by studying the original model (\ref{eqn1}) without the delay.

\begin{figure}[!ht]
	\begin{center}
		\subfigure[Cell Evolution with time]{%
			\label{fig_1a}
			\includegraphics[height=2in,width=3in]{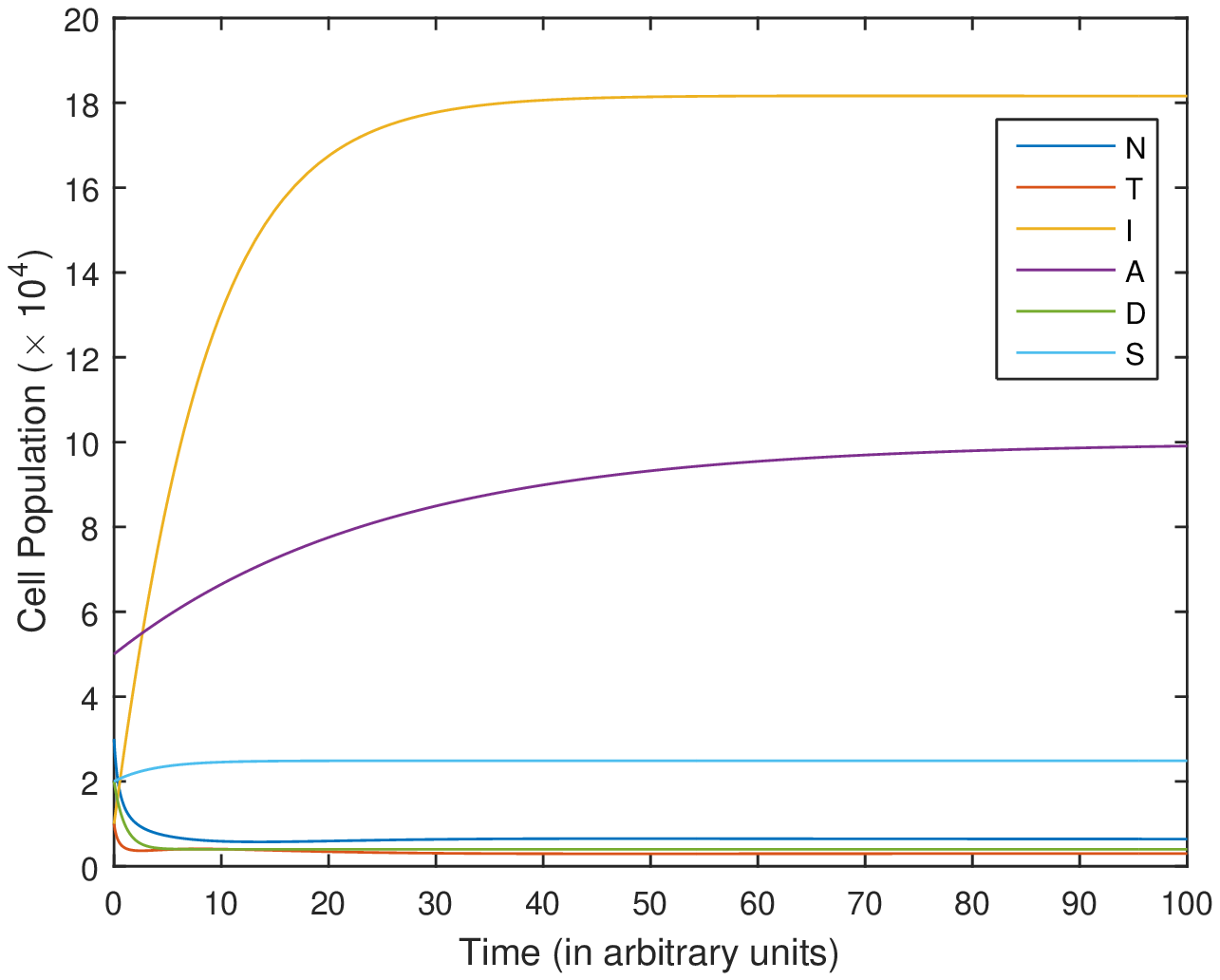}            
		}
		\subfigure[Evolution of Tumor cells with various $\beta$]{%
			\label{fig_1b}
			\includegraphics[height=2in,width=3in]{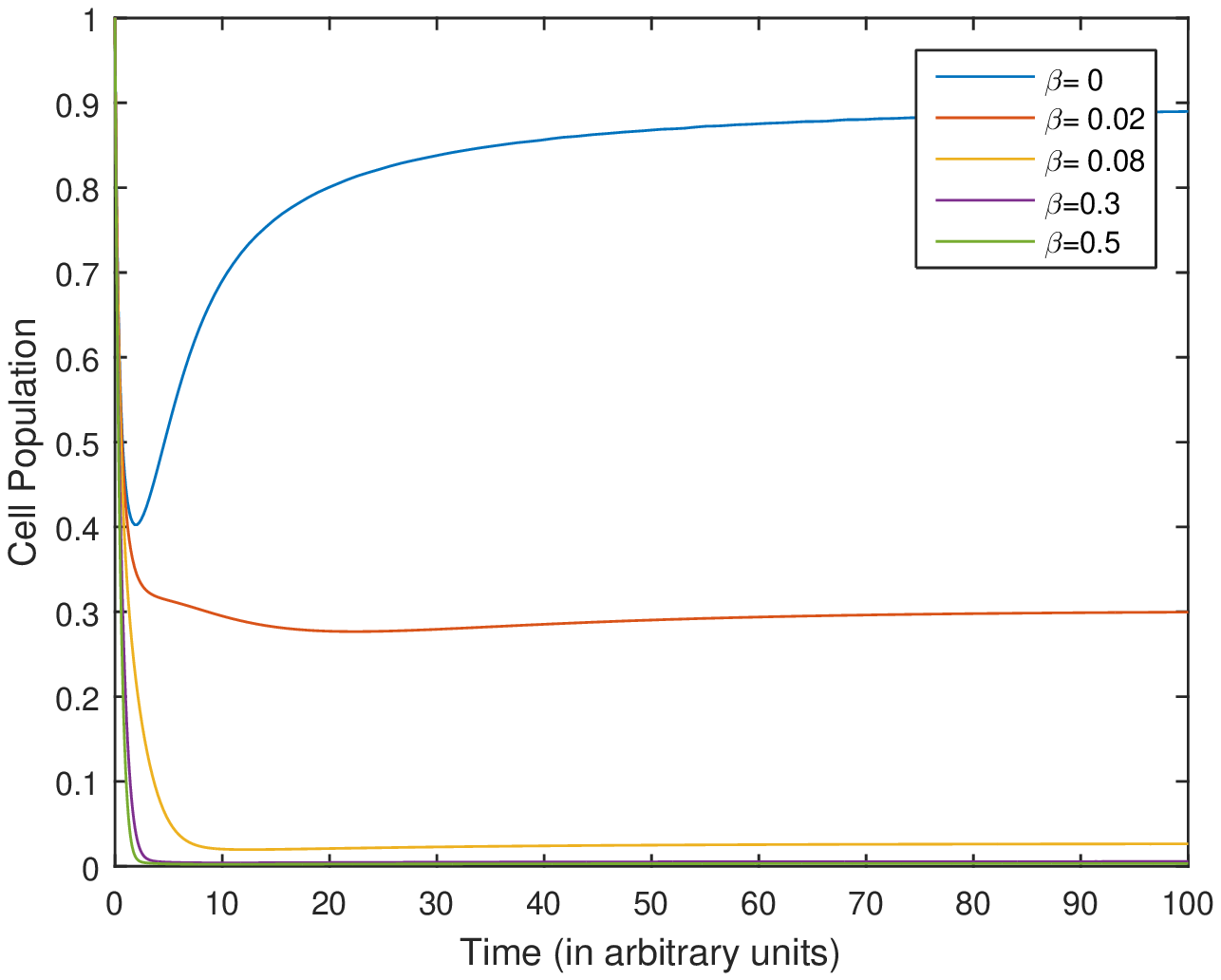}         
		}%
	\end{center}
\caption{Evolution of cells for (a) parameters given in (\ref{parm}) and (b) for various $\beta$}
\end{figure}

\begin{figure}[!ht]
	\begin{center}
		\subfigure[Evolution of Tumor cells with various $\alpha_2$]{%
			\label{fig_2a}
			\includegraphics[height=2in,width=3in]{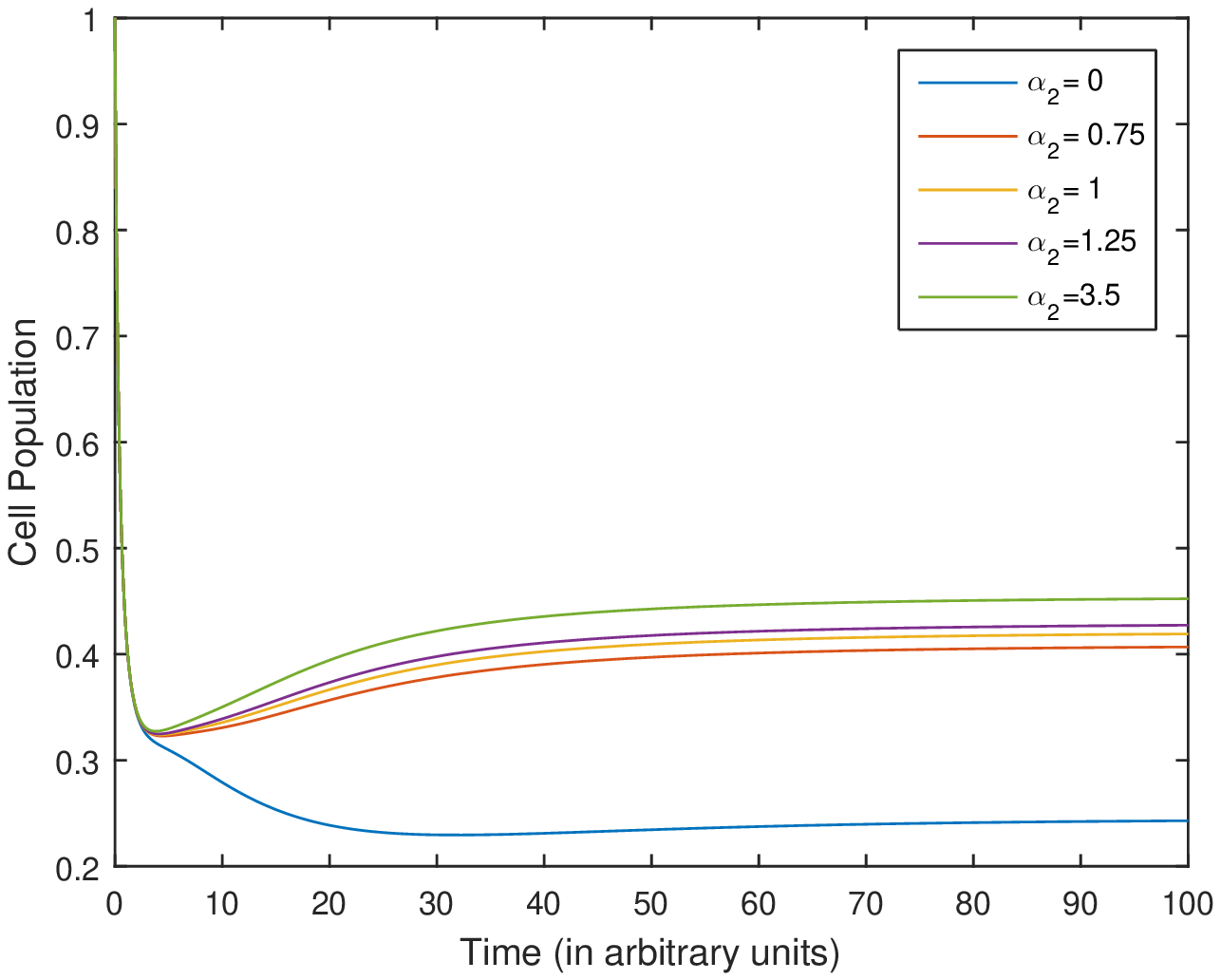}           
		}
		\subfigure[$R_0$ for various $\alpha_2$ and $\delta$]{%
			\label{fig_2b}
			\includegraphics[height=2in,width=3in]{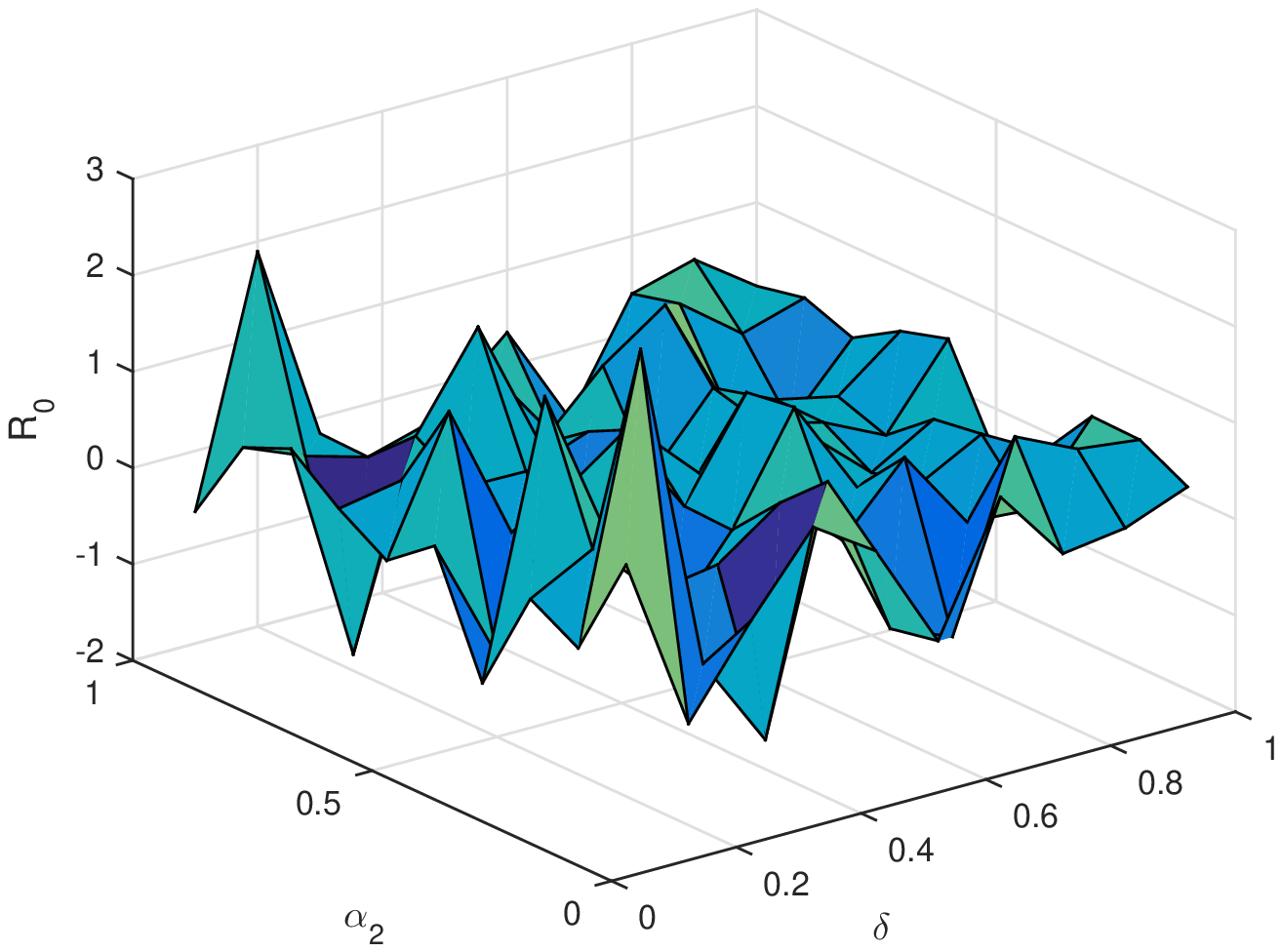}
		}%
	\end{center}
\caption{Evolution of cells for (a) various $\alpha_2$ and (b) $R_0$ for various $\alpha_2$ and $\delta$}
\end{figure}

For the numerical simulation of our model (\ref{eqn1}), we consider the parameter values as,
\begin{align}\label{parm}
& r_1 = 1, r_2 = 0.5, s_0 = 0.33, b_1 = 1, b_2 = 1, c_1 = 0.5,  c1' = 0.5, \alpha = 0.003, \alpha' = 0.003, \beta = 0.02,\\ &\beta' = 0.02, \gamma = 0.2, d_I = 0.2, d_D = 1, d_A = 0.04, 
\eta = 0.5,  \mu_1 = 0.3,  b = 0.5, \rho = 0.01, \alpha_1 = 0.3, \\&\delta = 0.08,  \alpha_2 = 0.1, b_0 =0.4, a_0 = 0.4. 
\end{align} 
Corresponding to these parameter values, the equilibrium points are $ P_1= (0,0,18.2923,10,0.4,2.4877), $ $ P_2 =(0,0.7403,8.9826,10,0.4,2.4878) $ and $P_3= (0.4170,0.5230,10.6200,10,0.4,2.4877). $
For this set of parameters, the reproduction number is $ R_0 = 0.3526. $  The simulation results for the model system (\ref{eqn1}) corresponding to these parameter values and initial population (30000, 10000, 10000, 50000, 20000, 20000) are shown in Figure \ref{fig_1a}.

It can be observed from the time portrait that the system is asymptotically stable, with the solutions converging to the equilibrium point $ P_3. $ Furthermore, we observe the evolution of Tumor cells with variations in parameters $ \beta $ and BT-immune threshold rate $ \alpha_2. $ We can observe from Figure \ref{fig_1b} that as the competition term $ \beta $ increases, population of Tumor cells decreases. This is due to competition between the Tumor cells and immune cells. 


\begin{figure}[!ht]
	\begin{center}
		\subfigure[Evolution of tumor cells with various interaction terms. ]{%
			\label{T_evol}
			\includegraphics[height=2in,width=3in]{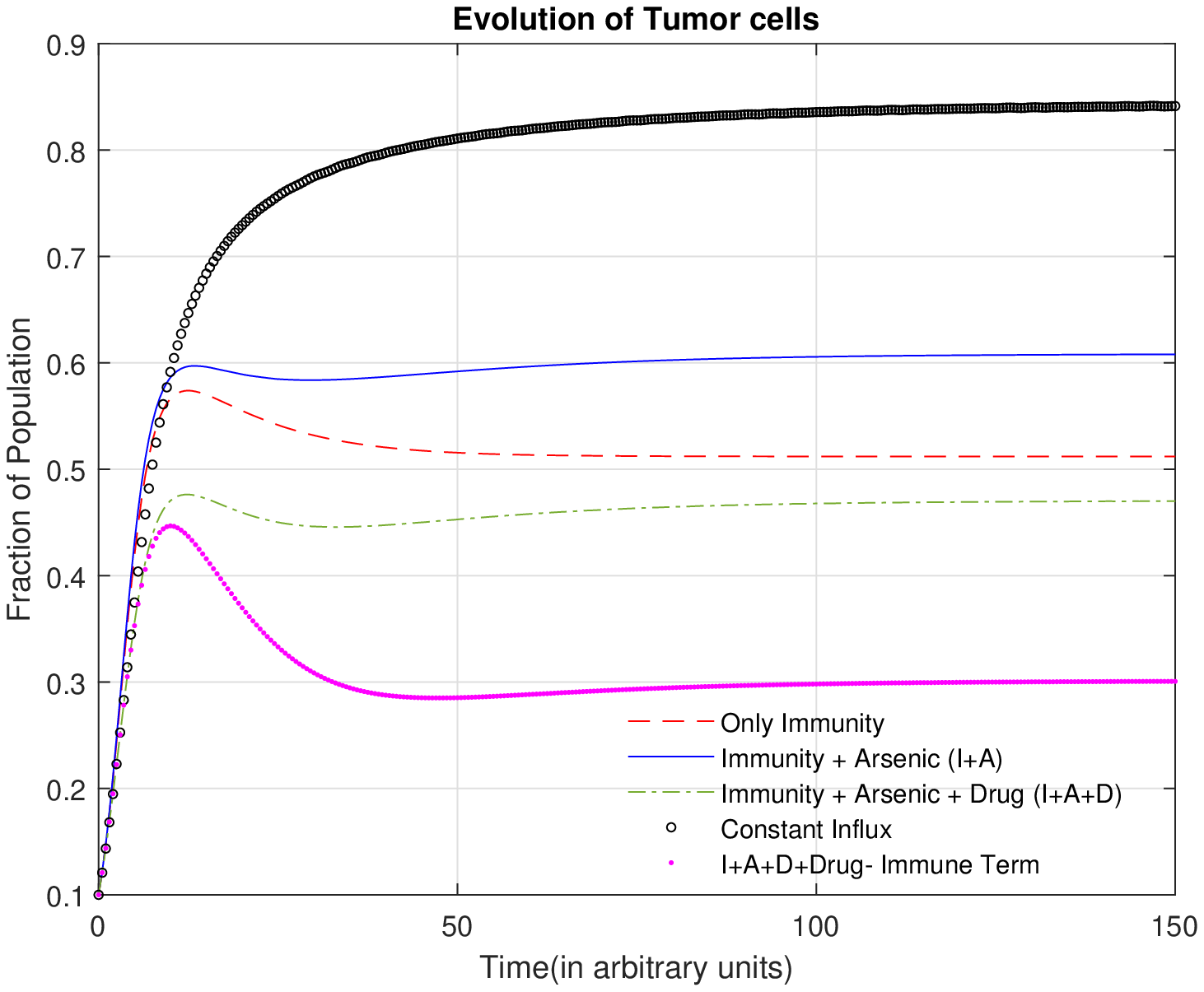}           
		}
		\subfigure[Evolution of normal cells with various interaction terms.]{%
			\label{N_evol}
			\includegraphics[height=2in,width=3in]{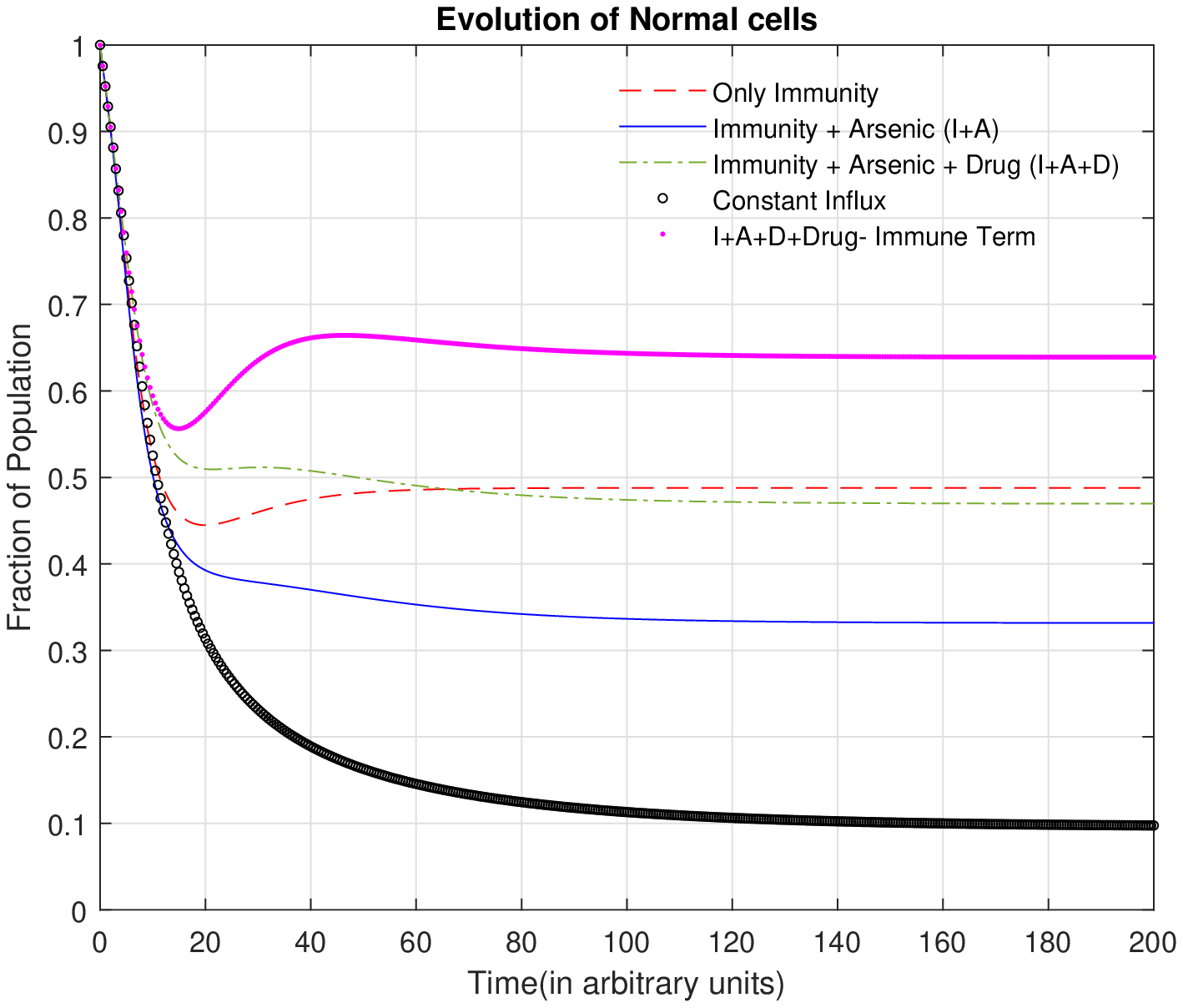}
		}%
	\end{center}
	\caption{Solution for variable-influx with BT-immune term converges to equilibrium state T*= 0.3005 for tumors and N*= 0.6399 for normal cells }
\end{figure}

\begin{figure}[!ht]
	\begin{center}
		\subfigure[Evolution of immune cells with various interaction terms]{%
			\label{I_evol}
			\includegraphics[height=2in,width=3in]{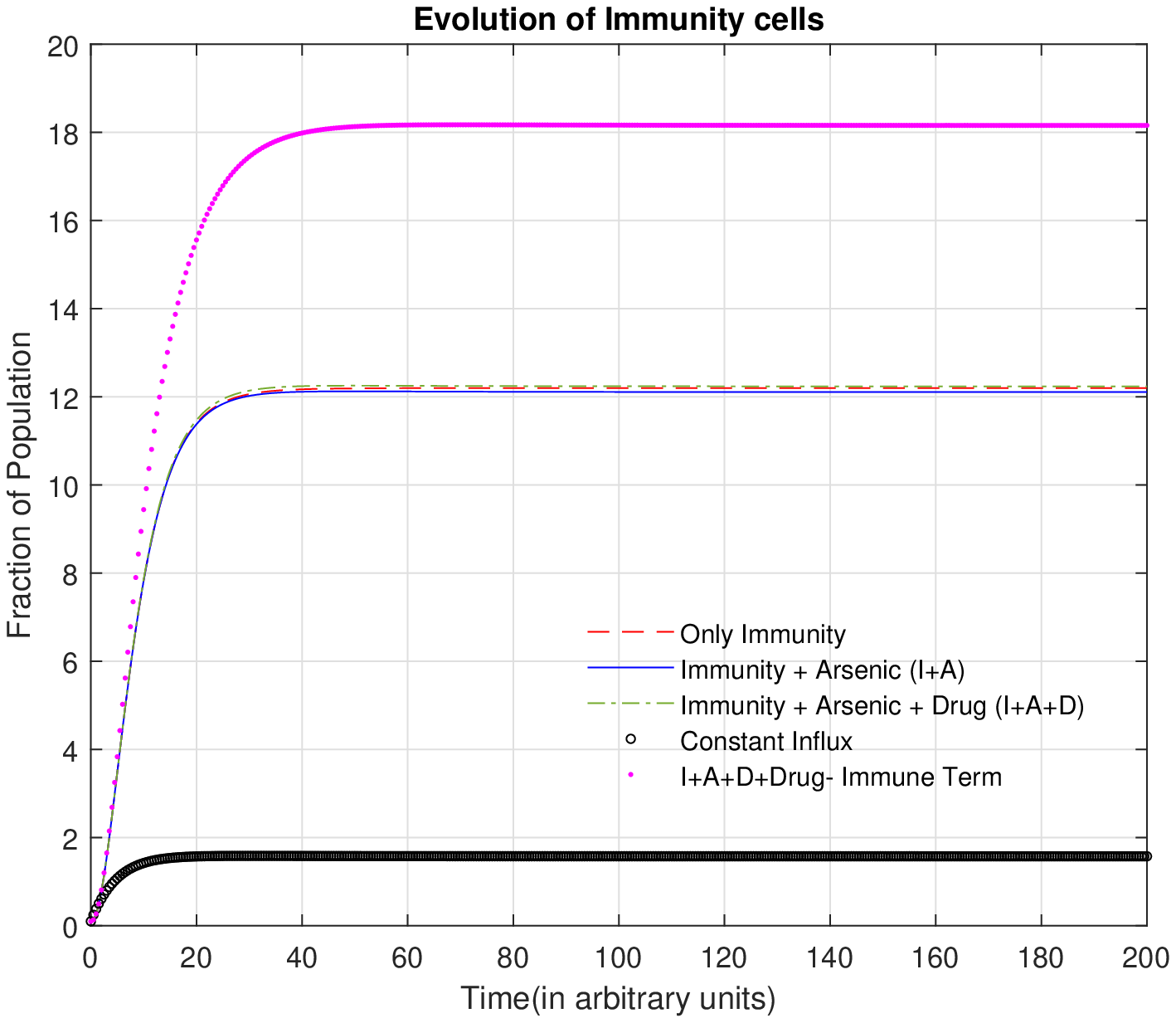}           
		}
		\subfigure[Evolution of cells under steady and variable influx.]{%
			\label{Steady_vs_Var}
			\includegraphics[height=2in,width=3in]{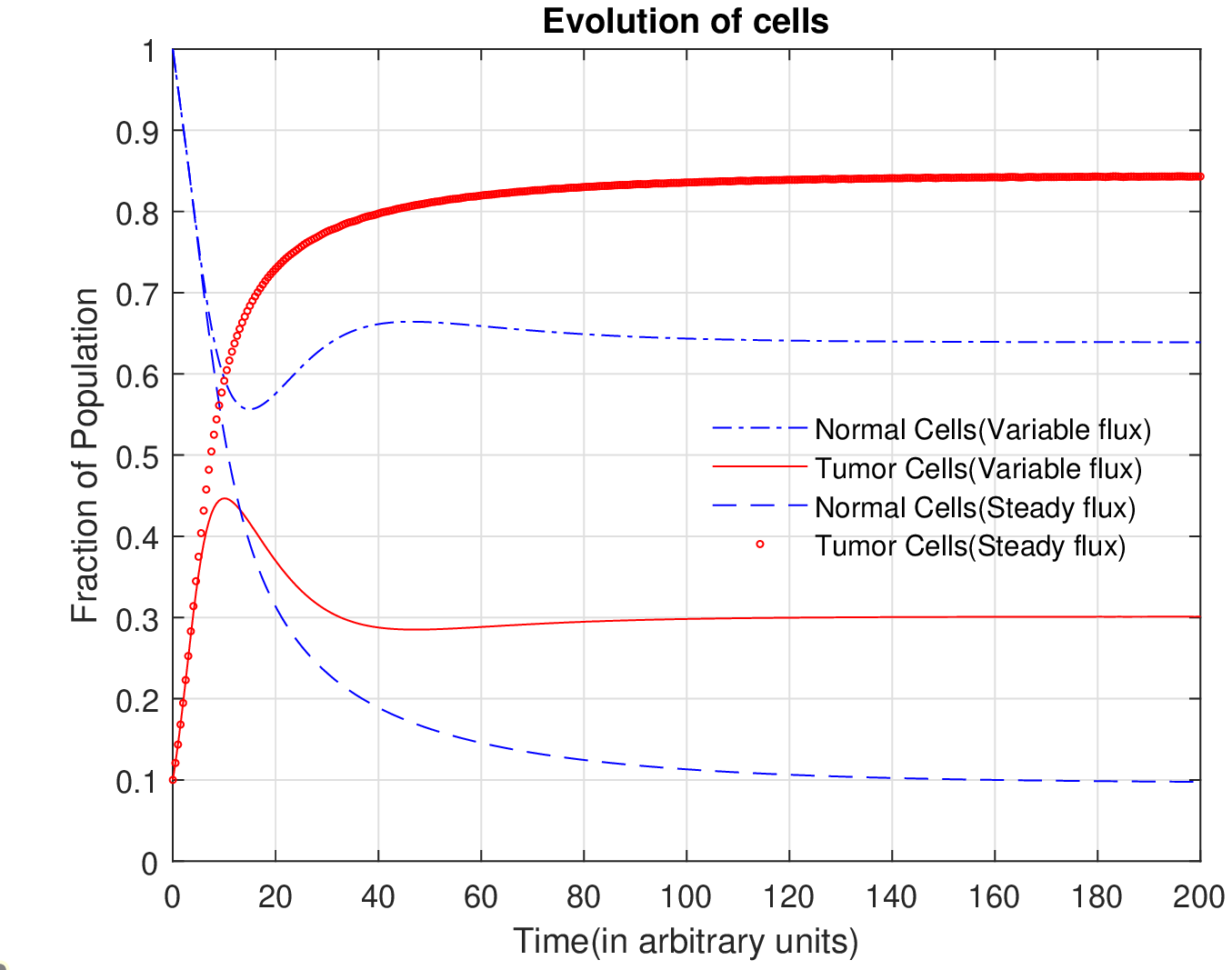}
		}%
	\end{center}
	\caption{Solution for steady immune converges to (N*, T*)= (0.0976,0.8431) and for variable influx it converges to (N*, T*)= (0.6390, 0.3010) }
\end{figure}

\begin{figure}[!ht]
	\begin{center}
		\subfigure[Evolution of tumor cells with variation in $\mu_1$]{%
			\label{N_T_var_mu1}
			\includegraphics[height=2in,width=3in]{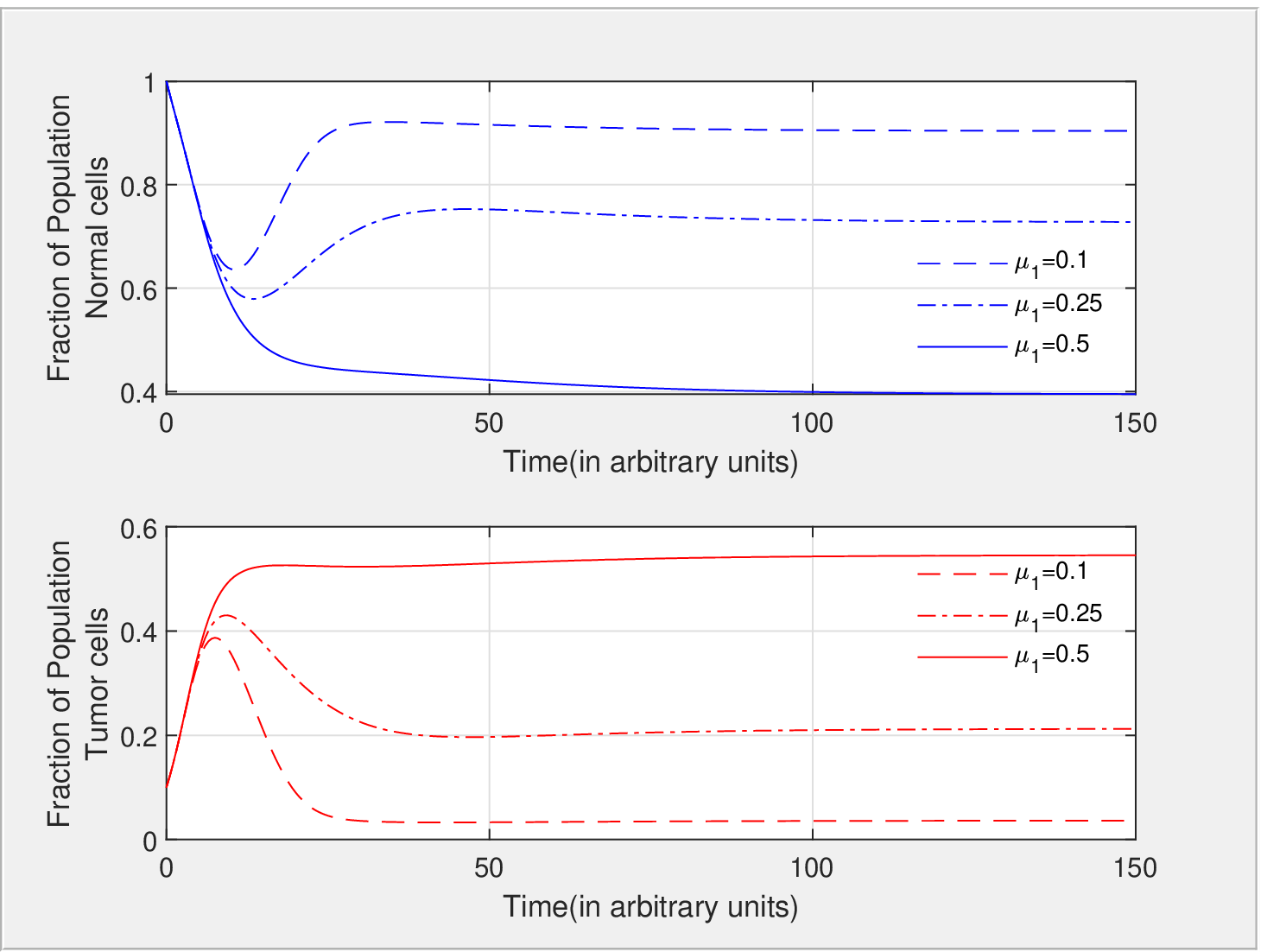}           
		}
		\subfigure[Evolution of normal cells with variation in $\eta$]{%
			\label{N_T_var_eta}
			\includegraphics[height=2in,width=3in]{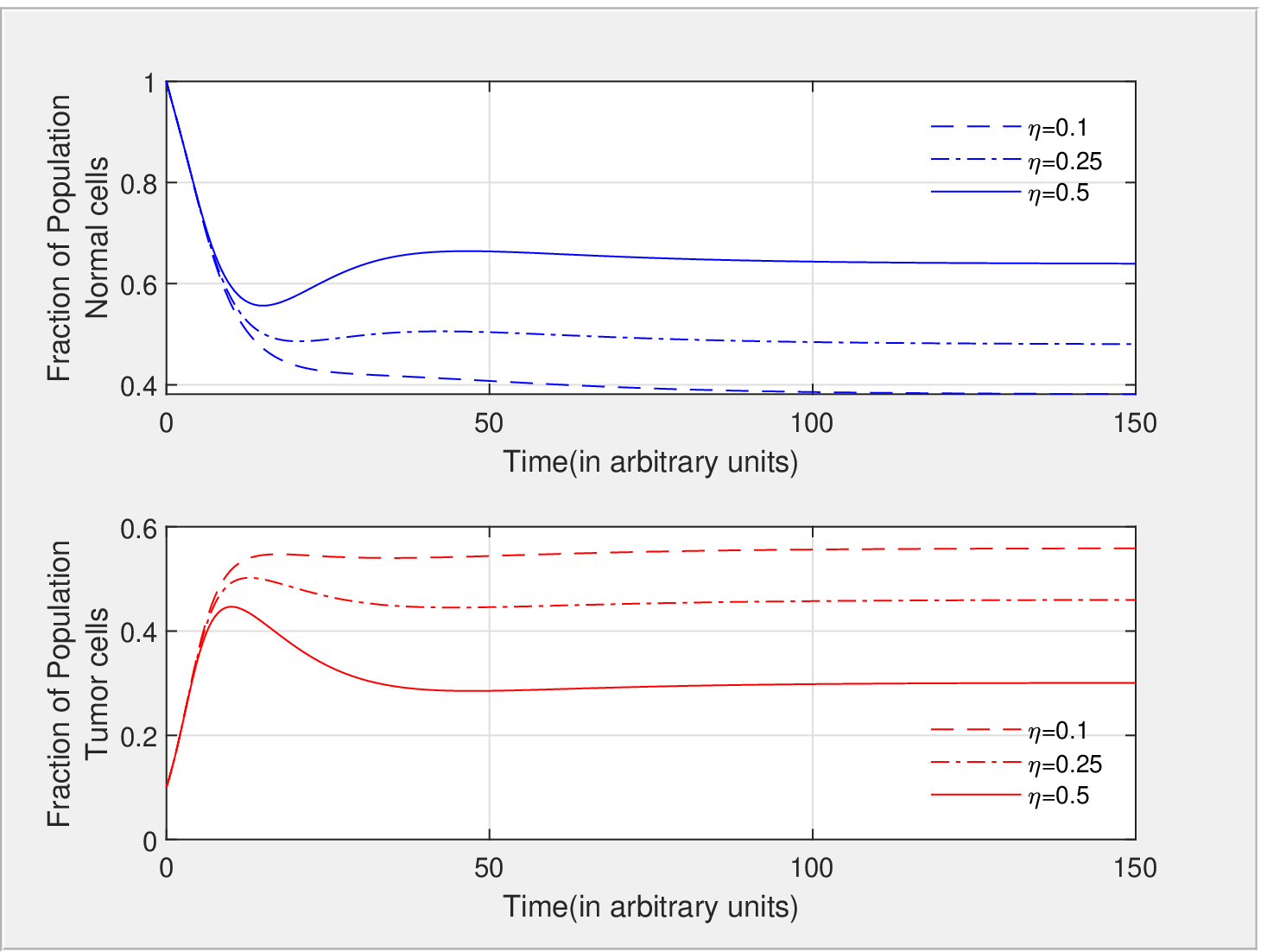}
		}%
	\end{center}
	\caption{Evolution of tumor and normal cells for various $\mu_1$ and $\eta$}
\end{figure}

\begin{figure}[!ht]
	\begin{center}
		\subfigure[]{%
			\label{phase_port_ANT}
			\includegraphics[height=3in,width=4in]{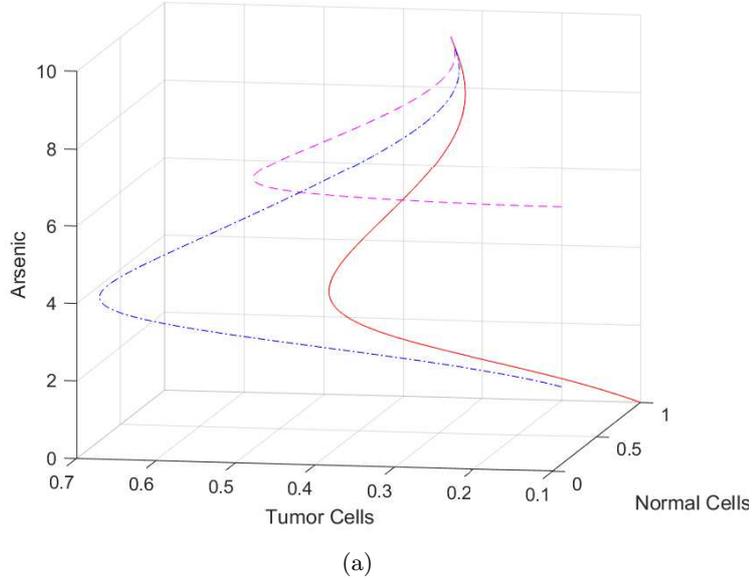}           
		}
	\end{center}
	\caption{Phase portrait for model \ref{eqn1}: Solution trajectories converge to its equilibrium state (N*, T*, A*)= (0.6399, 0.3005, 9.961) starting from different initial conditions.}
\end{figure}

The behaviours of tumor, normal and immune cells are studied using similar parameter set as before. In figure \ref{T_evol} and \ref{N_evol}, the behaviours of the tumor and normal cells are illustrated in the presence of various interaction components.

As the system of tumor cells interacts with immune cells in steady influx model (black circles in fig.\ref{T_evol} ), the saturation value is high enough to consider the system dead. The population is simulataneously low for the normal cells as shown in fig. \ref{N_evol}. This gives us the idea to supplement the immune cell growth to counteract the rise of tumor cells. A variable influx is then assumed in system with immunity alone, and the pronounced effect of variable influx is immediately observed as the tumor population comes down, while increasing the normal cells population (red dashed line in fig. \ref{T_evol} and \ref{N_evol}). Cancer inducing arsenic is then introduced in the system with the variable influx model and the tumor population is raised, and decreasing the normal population. The is due to the fact that arsenic converts a fraction of normal cells into tumorous. Black tea is introduced in the system in form drug  and while this does affect the population of tumor and normal cells (green dash-dot in the figure), the introduction of BT-immune interaction has great impact in the system as shown in fig.\ref{T_evol}(magenta dots).

For variable influx model we observe that immune cells saturate at a higher value than for the steady influx and the effect is even pronounced if BT-immune term is incorporated. This indicates that it may be possible to reach a healthy outcome with an immunomodulator and variable immune response, working in tandem, without the intervention of chemotherapeutic drugs. Indeed, it depends on the values of the parameters and the severity of the disease, but what we are able to show here is that it may be possible to dispense with chemotherapy or at least reduce its application to a large degree if suitable protocols with BT and immune response could be achieved. 

While the perfect cure for cancer is still not possible, we can aim to provide a good and prolonged life. One way to achieve that is by keeping the tumor cells always under control and lower than normal cell population. This can be achieved by extensive chemotherapy but the price is paid in form of painful side-effects. Starting with a small tumor population, and supplementing the system with immunomodulatory effects of BT, we observe as shown in fig.\ref{Steady_vs_Var}, that tumor population can be kept below a certain threshold for a longer amount of time.

The variation in the immune cells in system (\ref{eqn1})is brought about by s(t). These stimulated immune cells are assumed to have production rate of $\eta$ and death rate of $\mu_1$. Higher rate of production of immune cells has larger effect on tumor population as shown in fig.\ref{N_T_var_mu1}. One has to be careful here as large production rate can lead to immune cell proliferation and immune upon immune crowding. A similar scenerio is for death of stimulated immune cells (see fig.\ref{N_T_var_eta}) and a smaller death rate can lead to similar problems as that of larger production rate.

Considering the parameteric values in sec. \ref{sec: non-delay}, the model system(\ref{eqn1}) has the equilibrium state solution (N*,T*,A*)= (0.6399, 0.3005, 9.961), and the solution curves stabilize to its equilibrium state (shown in phase portrait \ref{phase_port_ANT} )  

\subsection{Delayed Model}

In this section we present the results obtained for a delayed system (\ref{eqn_delay}). As discussed in Section \ref{sec_Delayed}, time delay of $\tau_1$ and $\tau_2$ is incorporated in the normal and tumor cells respectively. This delay could be due immunity in the system .

A delay $\tau_1=\tau_2=\tau$ has been added in Arsenic-Normal cell interaction. This kind of delay in the system can be analysed analytically as done in Section \ref{sec_Delayed}. Physically this delay implies that the arsenic in the system converts the normal cell into tumorous at a later time. The effect of the delay doesn't have a large impact over the fraction of cells as shown in fig \ref{N_T_delay}. For a delay of $\tau=5$ the normal cells fare better very slightly before converging to its stable equilibrium point (see fig \ref{phase_port_delay})

A more physical delay in the form of delay in immune response is also studied. While the system is complex, the analytic research about the characteristic equation can be carried out similarly, and we present the numerical analysis in this section to show the dynamical behaviour of the system with delayed immune response.

\begin{figure}[!ht]
	\begin{center}
		\subfigure[Solution for model \ref{eqn_delay}. Solution converges to its equilibrium state (N*, T*)= (0.6439, 0.3005)]{%
			\label{N_T_delay}
			\includegraphics[height=2in,width=3in]{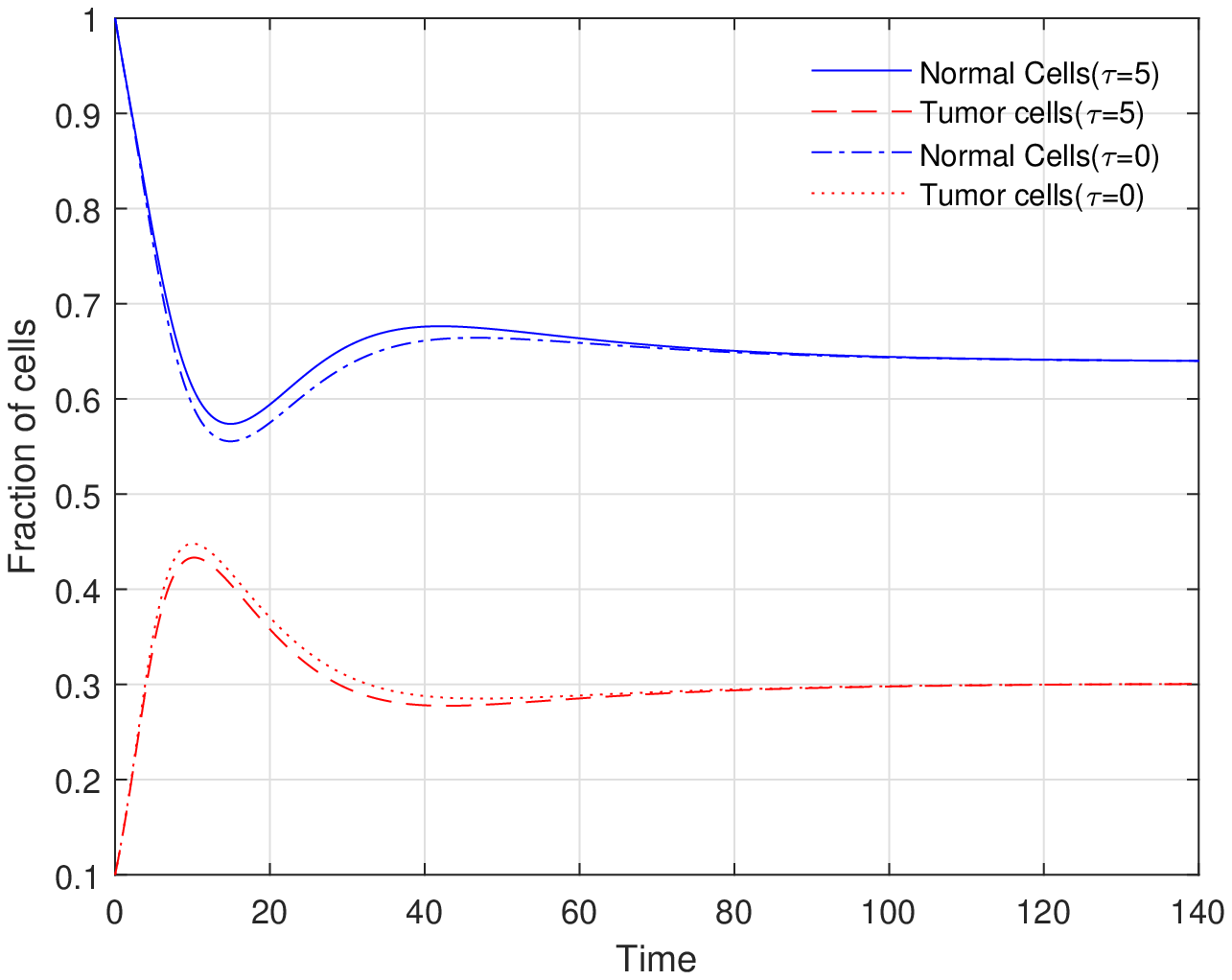}           
		}
		\subfigure[Phase portrait for model \ref{eqn_delay}. Solution converges to its equilibrium state (N*, T*)= (0.6439, 0.3005)]{%
			\label{phase_port_delay}
			\includegraphics[height=2in,width=3in]{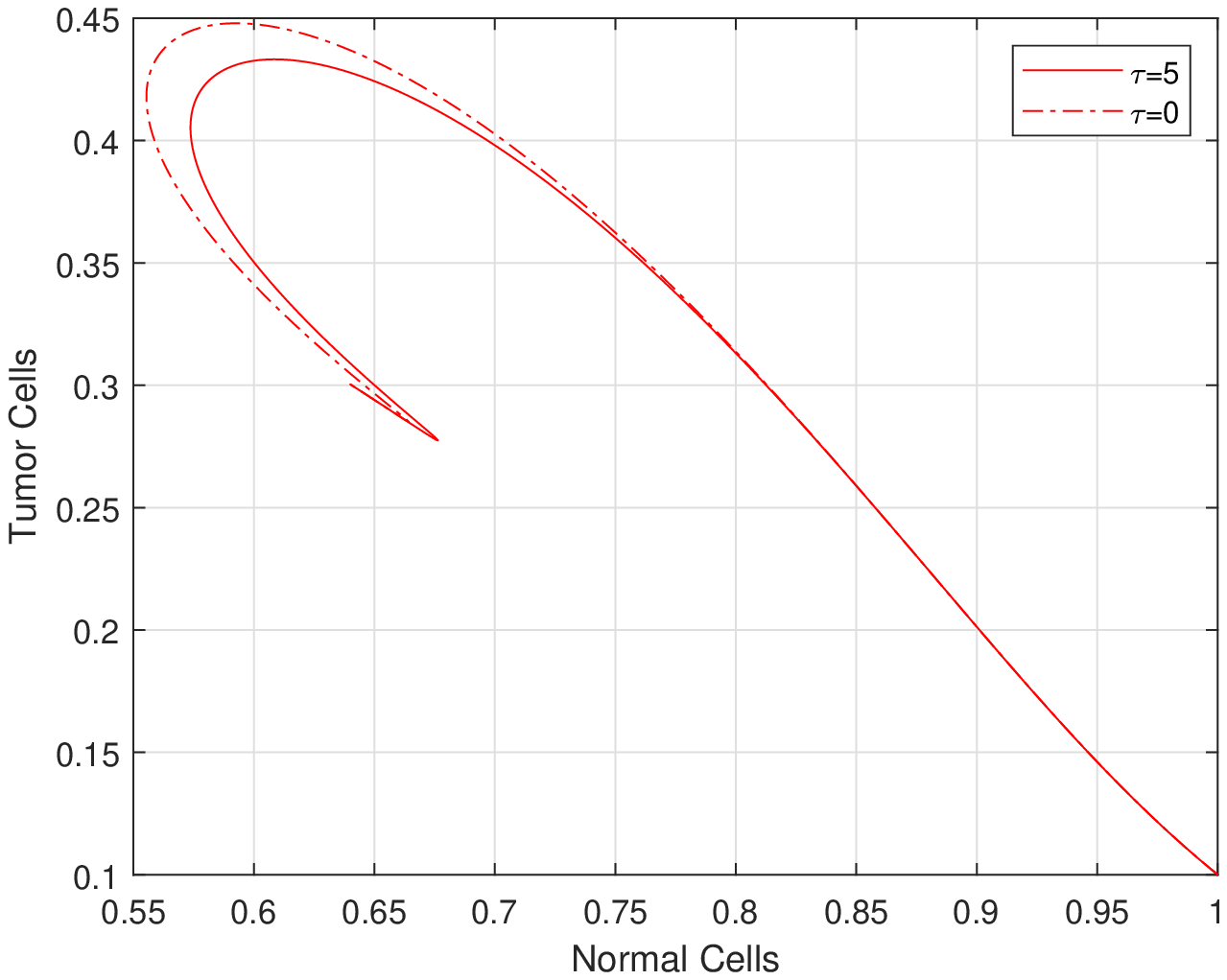}
		}%
	\end{center}
	\caption{Solution converges to its equilibrium state (N*, T*)= (0.6439, 0.3005)}
\end{figure}

\begin{figure}[!ht]
	\begin{center}
		\subfigure[Solution for model \ref{eqn_delay_immune}. Solution converges to its equilibrium state (N*, T*)= (0.6439, 0.3005)]{%
			\label{N_T_delay_immune}
			\includegraphics[height=2in,width=3in]{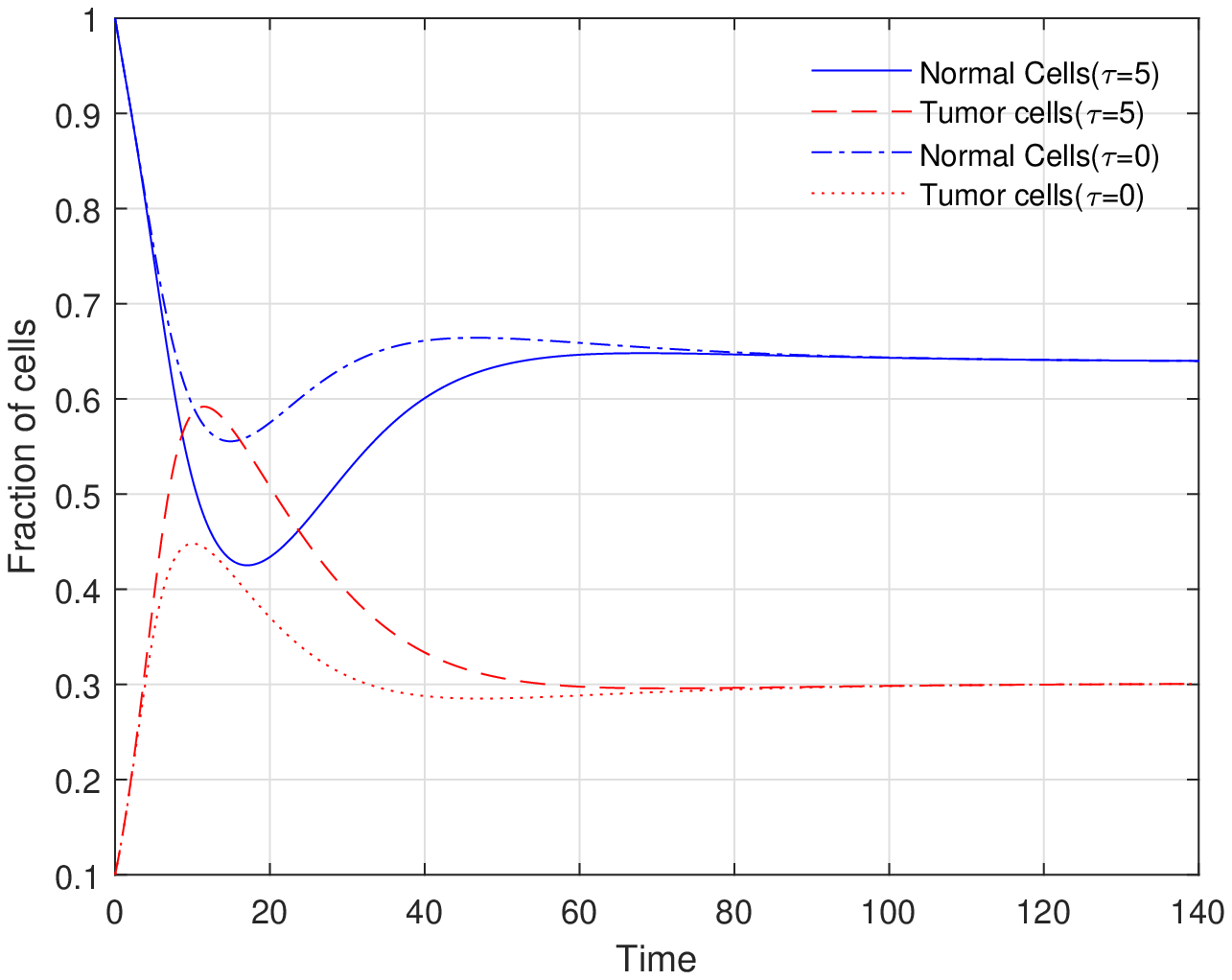}           
		}
		\subfigure[Phase portrait for model \ref{eqn_delay_immune}. Solution converges to its equilibrium state (N*, T*)= (0.6439, 0.3005)]{%
			\label{phase_port_delay_immune}
			\includegraphics[height=2in,width=3in]{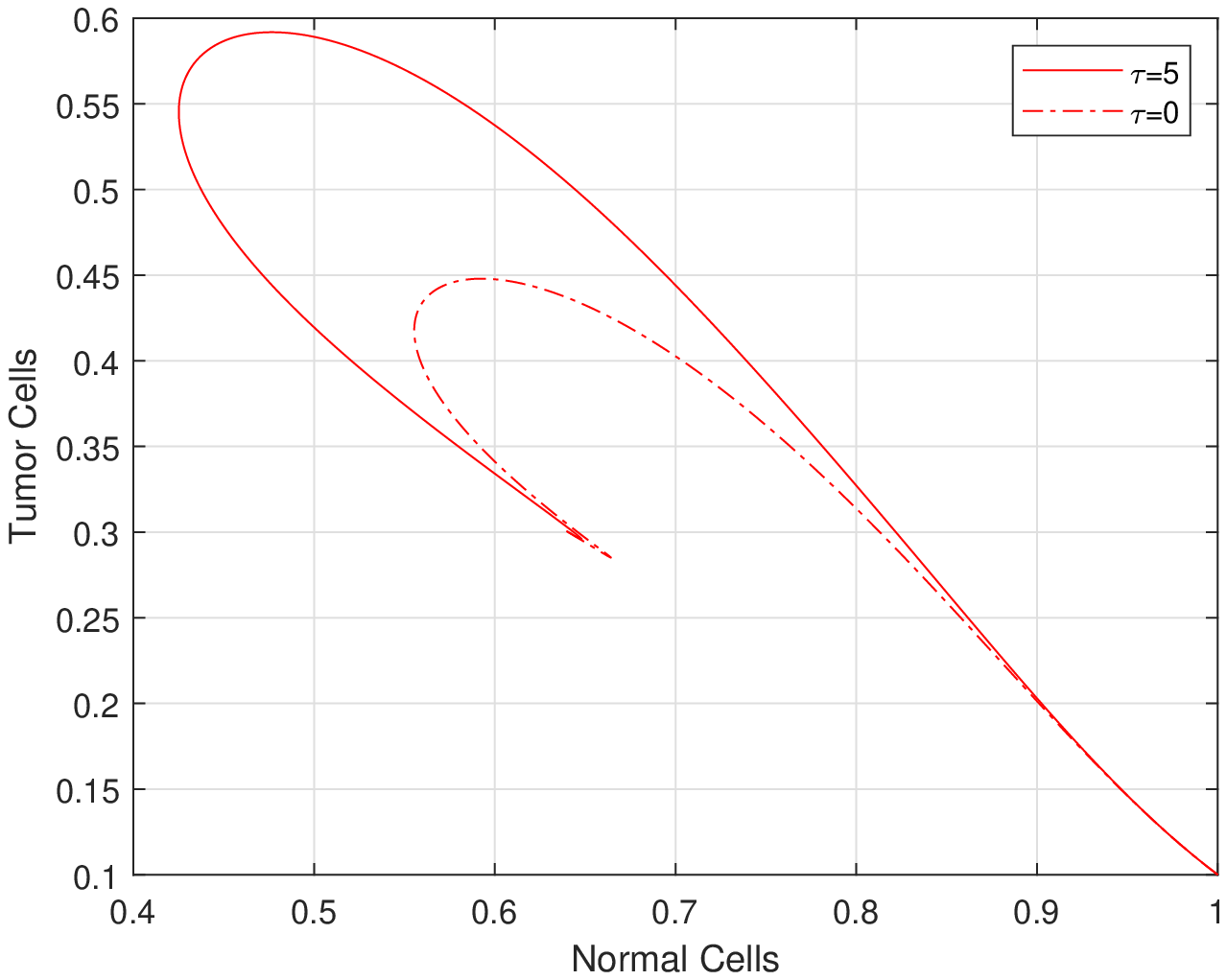}
		}%
	\end{center}
	\caption{Solution converges to its equilibrium state (N*, T*)= (0.6439, 0.3005)}
\end{figure}

\begin{eqnarray} 
\frac{dN}{dt} &=& r_{2}N\left(1-b_{2}N\right)-c_{1}TN-\alpha' AN\nonumber\\
\frac{dT}{dt} &=& r_{1}T\left(1-b_{1}T\right)-c_{1}' TN-\beta I(t-\tau)T-\gamma DT+\alpha AN\nonumber\\
\frac{dI}{dt} &=& s(t)+\frac{\rho I(t-\tau)T}{\alpha_1 +T}-d_{I}I-\beta I(t-\tau)T+\frac{\delta I(t-\tau)D}{\alpha_{2}+D}\label{eqn_delay_immune}\\
\frac{dA}{dt} &=& a_0-d_{A}A\nonumber\\
\frac{dD}{dt} &=& b_0-d_{D}D\nonumber\\
\frac{ds}{dt} &=& s_{0}+\frac{\eta}{b+s}s-\mu_{1}s.\nonumber
\end{eqnarray}

A delay in tumor-immune and  BT-immune response is studied in eqn.\ref{eqn_delay_immune}. A delay of $\tau=5$ has considerable effect on fraction of cells as shown in fig \ref{N_T_delay_immune}. The fraction of tumor cells rises above the normal cells before converging towards equilibrium point (N*, T*)= (0.6439, 0.3005), as observed in Fig. \ref{phase_port_delay_immune}.

\section{Conclusion}

We have written down a model for the growth of cancer cells under the exposure of an environmental carcinogen and a mitigating agent. We deliberately exclude any chemotherapeutic treatment in the model in order to test and establish a possible protocol without its use. The treatment of the model follows two routes, one analytical looking for the local and global stabilities of dead equilibrium and endemic equilibrium from the model.The Reproduction Number, which is an indicator of rate of spread has been computed by using next generation matrix method. The local and global stability of dead equilibrium and endemic equilibrium have been analyzed. Furthermore, stability and bifurcation analysis of delayed model have been discussed. 

A model system with compromised immune system has been studied with prime focus on using immunotherapy to obtain best results within certain general parameters. A comparison has been made between a constant influx model of immune cells studied earlier and the variable influx model of immune cells under consideration. Black tea, which has been studied extensively in literature is used as an immunomodulator to curb the side-effects of traditional chemotherapeutic treatments. Along with variable influx model an important BT-immune interaction has been incorporated in the system. Constant influx of immune cells saturates the tumor cells at unhealthy levels even under BT-immune interaction. The maximum effect on tumor cells comes from combining the two predominant factors, namely, variable influx and BT-immune interaction. Starting from as high as 10 percent tumor cell count to begin with, the tumor cell population never exceeds the normal cells at any time during the evolution. The effect of rate of production of immune cells and their death have been investigatyed and a higher production rate, and slower death rate of immune cells have greater effect on tumor cells; however this has to be exercised with caution as this may lead to immune proliferation and immune upon immune crowding.

Finally, a physically relevant scenario with delay has been introduced. A delay in arsenic-normal cell interaction some impact on the general result, albeit small quantitative shifts in counts. On the other hand, a delay in immune cell response has considerable effect on the fractional count of cells. A delay of $\tau=5$ results in higher peak value of tumor cells and even surpasses the fraction of normal cells highlighting the importance of rapid immune response in addition to the variable immune influx in fighting cancer.

\bibliographystyle{unsrt}
\bibliography{Cancer_paper}
%
%
%

\end{document}